\documentclass[11pt]{article}
\usepackage[utf8]{inputenc}
\usepackage{amsthm, amssymb, amsmath, graphicx, xcolor}
\usepackage[margin = 1in]{geometry} 

\newtheorem{theorem}{Theorem}[section]
\newtheorem{lemma}[theorem]{Lemma}
\usepackage[pagewise]{lineno}

\title{The single-face ideal orientation problem in planar graphs}
\author{Yipu Wang}

\begin{document}

\maketitle

\begin{abstract}
    We consider the ideal orientation problem in planar graphs. In this problem, we are given an undirected graph $G$ with positive edge lengths and $k$ pairs of distinct vertices $(s_1, t_1), \dots, (s_k, t_k)$ called terminals, and we want to assign an orientation to each edge such that for all $i$ the distance from $s_i$ to $t_i$ is preserved or report that no such orientation exists. We show that the problem is NP-hard in planar graphs. On the other hand, we show that the problem is polynomial-time solvable in planar graphs when $k$ is fixed, the vertices $s_1, t_1, \dots, s_k, t_k$ are all on the same face, and no two of terminal pairs cross (a pair $(s_i, t_i)$ crosses $(s_j, t_j)$ if the cyclic order of the vertices is $s_i,s_j,t_i,t_j$). For serial instances, we give a simpler and faster algorithm running in $O(n \log n)$ time, even if $k$ is part of the input. (An instance is serial if the terminals appear in cyclic order $u_1, v_1, \dots, u_k, v_k$, where for each $i$ we have either $(u_i, v_i) = (s_i, t_i)$ or $(u_i, v_i) = (t_i, s_i)$.) Finally, we consider a generalization of the problem in which the sum of the distances from $s_i$ to $t_i$ is to be minimized; in this case we give an algorithm for serial instances running in $O(kn^5)$ time.
\end{abstract}

\section{Introduction}
Let $G$ be an undirected graph, and let $(s_1, t_1), \dots, (s_k, t_k)$ be $k$ pairs of vertices in $G$.
An {\em orientation} of $G$ is a directed graph that is formed by assigning a direction to each edge in $G$.
In the {\em orientation problem}, we want to find an orientation of $G$ such that for all $i$, $s_i$ can reach $t_i$.
If for all $i$ we further require that the distance from $s_i$ to $t_i$ be preserved, then we get the {\em ideal orientation} problem. For some graphs $G$, an ideal orientation may not exist, and in such cases we may want to minimize the sum of the distances from $s_i$ to $t_i$; this gets us the {\em $k$-min-sum orientation problem}. 
Minimizing the longest distance gets us the {\em $k$-min-max orientation problem}, while minimizing the shortest distance gets us the {\em $k$-min-min orientation problem}.

Ito et al.~\cite{IMOTU09} suggest the following application of the orientation problem. Suppose we have to assign one-way restrictions to aisles in, say, an industrial factory, while maintaining reachability between several sites. This corresponds to the orientation problem. We may also want to maintain the distances of routes between the sites in order to keep transit time low and productivity high; this corresponds to the ideal orientation problem.

The orientation problem was first studied by Hassin and Megiddo~\cite{HM89}, and they gave the following algorithm that works in general graphs. Without loss of generality, assume that $G$ is connected. First, compute the bridges of $G$. (A bridge is an edge whose removal would disconnect $G$). For each $i$, pick an arbitrary path from $s_i$ to $t_i$, and orient the bridges on this path in the direction that they appear on this path. If a bridge is forced to be oriented in both directions, then no orientation preserving reachability exists. Otherwise, such an orientation does exist: in the rest of $G$ each component is a 2-connected component and can be oriented to be strongly connected, by Robbins' theorem~\cite{R39}. 

By contrast, much less is known about the ideal orientation problem and generalizations like the $k$-min-sum orientation problem. Hassin and Megiddo showed that the ideal orientation problem is polynomially-time solvable when $k=2$ but is NP-hard for general $k$. Eilam-Tzoreff~\cite{E-T98} extended Hassin and Megiddo's algorithm when $k=2$ to find an ideal orientation minimizing the number of shared arcs in the paths realizing the distances in $H$. She also solved the generalization when $k=2$ and we only require the shorter distance in $H$ to be a distance in $G$. The complexity of the ideal orientation problem for fixed $k > 2$ remains open.

Fenner, Lachish, and Popa~\cite{FLP13} considered the min-sum orientation problem in general graphs when $k=2$. They give a PTAS and reduce the 2-min-sum orientation problem to the 2-min-sum edge-disjoint paths problem. (In the 2-min-sum edge-disjoint paths problem, we need to find edge-disjoint paths from $s_1$ to $t_1$ and from $s_2$ to $t_2$ of minimum total length.) It remains unknown whether the 2-min-sum orientation problem or the 2-min-sum edge-disjoint paths problem can be solved in polynomial time.

Ito et al.~\cite{IMOTU09} considered the $k$-min-sum and $k$-min-max orientation problems. They proved that both problems are NP-hard in planar graphs, and that the $k$-min-sum orientation problem is solvable in $O(nk^2)$ time if $G$ is a cactus graph and $O(n + k^2)$ time if $G$ is a cycle. They showed that the $k$-min-max orientation problem is NP-hard in cacti, even when $k=2$, but solvable in cycles in $O(n + k^2)$ time. For the $k$-min-max orientation problem, they also give a 2-approximation in cacti and a fully polynomial-time approximation scheme for fixed $k$ in cacti. It remains an open question whether $k$-min-sum or $k$-min-max orientation problems can be solved or approximated in classes of graphs more general than cacti.

In this paper we present four results, three of which deal with the ideal orientation problem and one of which deals with the $k$-min-sum problem.
First, we solve the ideal orientation problem for planar instances for {\em serial} instances, even if $k$ is part of the input. An instance of any orientation problem is {\em serial} if the terminals are all on a single face in cyclic order $u_1, v_1, \dots, u_k, v_k$, where for each $i$ we have either $(u_i, v_i) = (s_i,t_i)$ or $(u_i,v_i) = (t_i, s_i)$. See Figure~\ref{F:serial}. The algorithm is simple and relies on the fact that we can assume that the paths realizing the $s_i$-to-$t_i$ distances are pairwise non-crossing.

\begin{figure}
    \centering
    	\includegraphics[scale = 0.25]{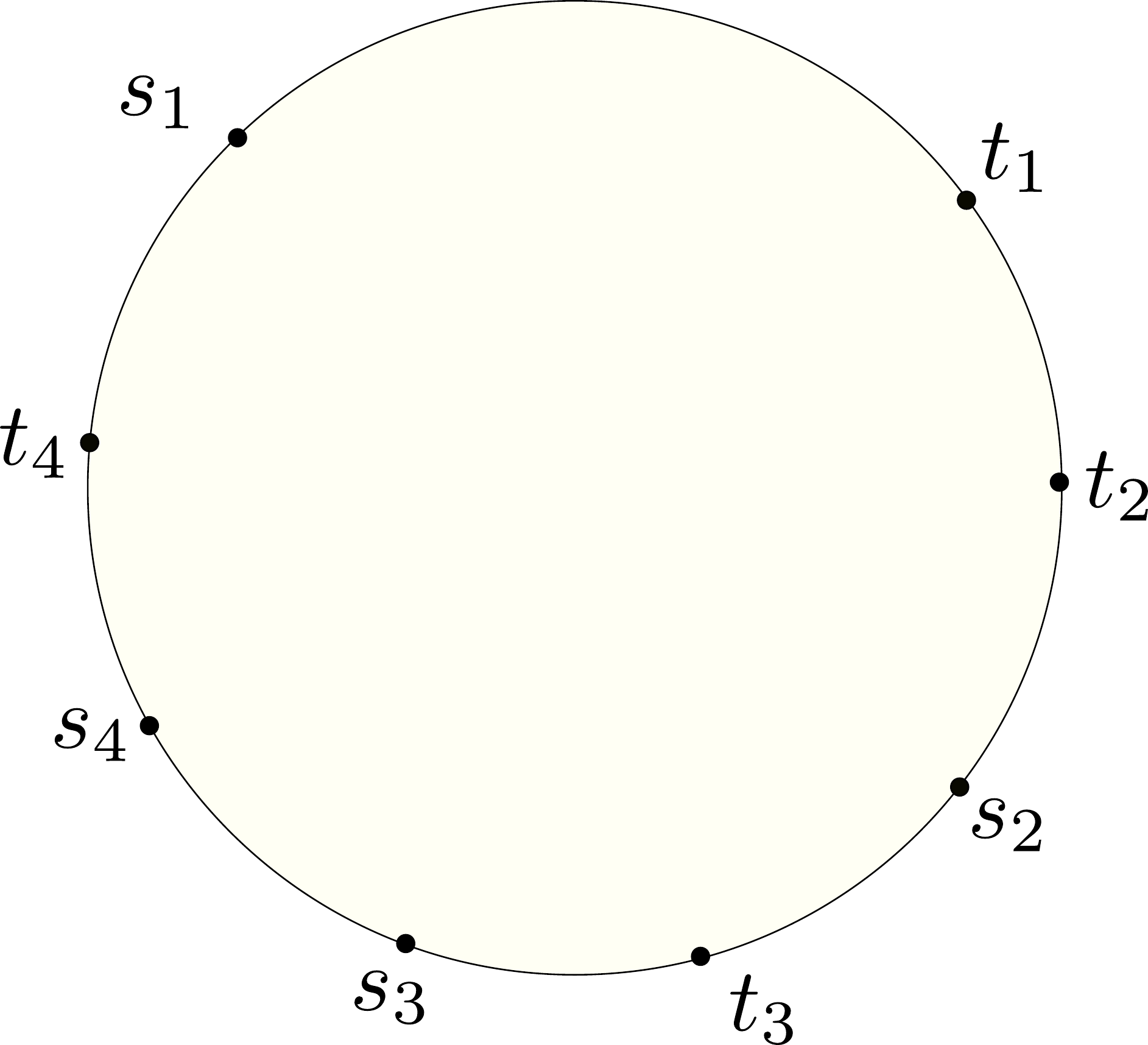} 
    \caption{a serial instance where $k=4$}
    \label{F:serial}
    \end{figure}

\begin{theorem}
    Any serial instance of the ideal orientation problem can be solved in $O(n \log n)$ time.
\end{theorem}
The algorithm uses Klein's MSSP algorithm~\cite{K05}, which computes an implicit representation of the solution. If an explicit orientation is desired, then a solution takes $O(n^2)$ time to compute.

Second, we solve the ideal orientation problem in planar graphs for a fixed number of terminals when all terminals are on a single face and no terminal pairs cross. Two pairs of terminals $(s_i, t_i)$ and $(s_j, t_j)$ cross if all four terminals are on a common face and the cyclic order of the terminals is $s_i, s_j, t_i, t_j$. The algorithm relies on an algorithm of Schrijver that finds partially vertex-disjoint paths in directed planar graphs~\cite{S15}.

\begin{theorem}
    If $k$ is fixed and all terminals are on the outer face and no terminals cross, then we can solve the ideal orientation problem in polynomial time.
\end{theorem}
It is likely that the algorithm of Theorem 1.2 can be generalized to the case where for each $i$, $s_i$ and $t_i$ appear on a common face $F_i$ (the faces $F_1, \dots, F_k$ need not be distinct), and the terminals are still non-crossing; however, the author has not yet verified the details. 
The restriction that the terminals be non-crossing may seem arbitrary, but can be motivated in the following way. Define the {\em demand graph} $G_D$ to be the graph with the same vertices as $G$ but with an edge $\{s_i, t_i\}$ for each $i$. Define $G + G_D$ to be the graph with the same vertices as $G$ (or $G_D$) and whose edge set is $E(G) \cup E(G_D)$. The case of non-crossing terminals is then exactly the case where $G + G_D$ is planar.

Third, we show that the ideal orientation problem is NP-hard in planar graphs. The reduction is from planar 3-SAT and is inspired by reductions by Middendorf and Pfeiffer~\cite{MP} and by Eilam-Tzoreff~\cite{E-T98}, who showed that finding disjoint paths and disjoint shortest paths are NP-hard in planar graphs.
Since the min-sum, min-max, and min-min orientation problem are all generalizations of the ideal orientation problem, this reduction shows that the min-sum, min-max, and min-min problems are also NP-hard. This is stronger than Ito et al.'s result because the ideal orientation problem is a special case of the $k$-min-sum orientation problem.
\begin{theorem}
    If $k$ is part of the input, then the ideal orientation problem is NP-hard in unweighted planar graphs.
\end{theorem}

Fourth, we solve the $k$-min-sum orientation problem for serial instances. To do this, we classify each terminal pair as clockwise or counterclockwise, and we break up the instance into two sub-instances, one of which consists only of clockwise pairs and the other of which consists only of counterclockwise pairs. It turns out that solving each sub-instance reduces to solving serial instances of a shortest vertex-disjoint paths problem, which can be done using an algorithm of Borradaile, Nayyeri, and Zafarani~\cite{BNZ15}. Finally, after solving the two sub-instances independently, we show that the two sub-solutions can be easily combined to solve the original instance.

\begin{theorem}
    Any serial instance of the $k$-min-sum orientation problem can be solved in $O(kn^5)$ time.
\end{theorem}

This paper is organized as follows. In section 2, we present definitions and a subroutine, and we reformulate the orientation problems in terms of {\em non-conflicting paths}. We will then use these reformulations in the rest of the paper. In section 3, we prove various structural results that we use in the rest of the paper. In section 4 we prove Theorem 1.1, in section 5 we prove Theorem 1.2, in section 6 we prove Theorem 1.3, and in section 7 we prove Theorem 1.4. The appendix contains details of proofs omitted in the main paper.

\section{Preliminaries}
Throughout this paper, $G$ is a simple undirected plane graph, each edge $e \in E(G)$ has a positive length $\ell(e) > 0$, and $(s_1, t_1), \dots, (s_k, t_k)$ are $k$ pairs of vertices in $G$.
The vertices $s_1, t_1, \dots, s_k, t_k$ are called {\em terminals}, and we assume that the $2k$ terminals are all distinct. (If two terminals, say $s_i$ and $s_j$, are not distinct, then we add new terminals $s_i'$ and $s_j'$ that will be terminals instead of $s_i$ and $s_j$, respectively, and we add new arcs $s_i's_i$ and $s_j's_j$. If $s_i$ and $s_j$ were on a common face then we can ensure $s_i'$ and $s_j'$ still are.) For most of this paper we will assume that all terminals are on a common face, which we assume is the outer face.
Let $n$ be the number of vertices of $G$, so that $G$ has $O(n)$ edges.
An {\em orientation} of $G$ is a directed graph $G'$ that is formed by replacing each edge $\{u,v\} \in E(G)$ with exactly one of the arcs $uv$ or $vu$. 
We write $\partial G$ to denote the outer face or boundary of $G$, and we write $\deg(v)$ to denote the degree of a vertex $v$.
If four terminals $s_i, t_i, s_j, t_j$ are on a common face, then we say that the terminal pairs {\em cross} if their cyclic order (either clockwise or counterclockwise) on the face is $s_i, s_j, t_i, t_j$; otherwise, the terminals are {\em non-crossing}.

A {\em directed walk} $P$ in $G$ is a sequence of arcs $(u_0, v_0), \dots, (u_p, v_p)$ such that $v_{i+1} = u_i$ for all $i \in \{0, \dots p-1\}$.
In a slight abuse of terminology, we will say that the walk $P$ {\em uses} the (undirected) edges $\{u_0, v_0\}, \dots, \{u_p, v_p\}$. The directed walk $P$ is in $G$ if $\{u_i, v_i\}$ is an edge in $G$ for all $i \in \{0, \dots p-1\}$.
If in addition $u_0, v_0, \dots v_p$ are distinct, then $P$ is a {\em directed path} in $G$.
The {\em reverse} of the walk $P$ is the directed walk $(v_p, u_p), \dots, (v_0, u_0)$ and is denoted by $rev(P)$.
If we let $e_i = \{u_i, v_i\}$ for all $i \in \{0, \dots, p-1\}$, then the length of $P$, denoted by $\ell(P)$, is $\sum_{i = 1}^{c-1} \ell(e_i)$.
The distance from a vertex $u$ to a vertex $v$ in a graph $G$ is the length of a shortest walk from $u$ to $v$ and is denoted by $d_G(u,v)$; this walk will be simple path. If $u$ appears before vertex $v$ on the walk $P$, then we write $u \prec_P v$ and use $P[u,v]$ to denote the subwalk of $P$ from $u$ to $v$; we will only use this notation when there is no risk of ambiguity. Two walks {\em touch} or {\em meet} if they share at least one vertex, and two regions {\em touch} or {\em meet} if their closures share at least one vertex. The concatenation of two walks $P$ and $Q$ is denoted $P \circ Q$.
We will sometimes treat paths as sets of edges.
In an abuse of terminology, we say that two directed paths are edge-disjoint if their underlying undirected paths are edge-disjoint (we assume each arc $uv$ is embedded together with its reverse $vu$). 

For any orientation $G'$ of $G$, let $d'(u,v)$ be the distance from $u$ to $v$ in $H$. In the {\em ideal orientation problem}, we want to find an orientation $G'$ of $G$ such that for all $i$, $d(s_i,t_i) = d'(s_i,t_i).$ 
It is possible to reformulate the ideal orientation problem in terms of finding {\em non-conflicting} shortest paths; we will use this reformulation in the rest of the paper.
Two directed walks $P$ and $Q$ in $G$ {\em conflict} if there is an edge $\{u,v\}$ in $G$ such that $uv$ is an arc in $P$ and $vu$ is an arc in $Q$.
Two walks are {\em non-conflicting} if they do not conflict.
The ideal orientation problem then asks us to find pairwise non-conflicting directed walks $P_1, \dots, P_k$ such that $P_i$ is a shortest path from $s_i$ to $t_i$ for all $i \in \{1, \dots k\}$.
We call the set of such paths a {\em solution} to the instance.
For some graphs $G$, a solution may not exist, and we may be interested in relaxing the requirement that each path in the solution be a shortest path. This motivates us to define the {\em $k$-min-sum orientation problem}, in which the input is the same as the input to the ideal orientation problem, and we still want to find paths $P_1, \dots, P_k$ such that $P_i$ connects $s_i$ to $t_i$, but now our goal is to minimize the sum of the lengths of the paths $P_i$ instead of insisting that each $P_i$ be a shortest path. Clearly, the $k$-min-sum orientation problem is at least as hard as the ideal orientation problem.

Our algorithms search for pairwise non-conflicting directed walks that are shortest paths connecting corresponding terminals, rather than explicitly seeking simple paths. Because all edge lengths are positive, the set of shortest walks will end up consisting of simple paths.
Note that given a directed walk $P$ that conflicts with itself, we can repeatedly remove directed cycles from $P$ to obtain a simple directed path $P'$ such that $P'$ has the same starting and ending vertices as $P$, $P'$ is no longer than $P$,  and $P'$ does not conflict with itself.
Thus we do not have to worry about directed walks conflicting with themselves.

We assume without loss of generality that the paths in any solution do not use edges on the outer face.  If necessary to enforce this assumption, we can connect the terminals using an outer cycle of $2k$ infinite-weight edges.

\subsection{Partially edge-disjoint non-crossing paths}\label{S:PNEPP}
Our algorithm for Theorem 2 ultimately involves two reductions. First, we reduce the ideal orientation problem (when terminals lie on a single face and are non-crossing) to a problem we call the {\em partially non-crossing edge-disjoint paths problem (PNEPP)}. This first reduction is described in Section~\ref{S:single-face}. Second, we reduce PNEPP to the {\em partially vertex-disjoint paths problem (PVPP)}. The goal of this subsection is to define both PNEPP and PVPP and to describe the second reduction. 

In PVPP, we are given a directed planar graph $H$, vertices $u_1, v_1, \dots, u_h, v_h$; subgraphs $H_1, \dots, H_h$ of $H$; and a set $S$ of pairs $\{i,j\}$ from $\{1, \dots, h\}$. We wish to find directed paths $Q_1, \dots, Q_h$ such that 
\begin{itemize}
    \item $Q_i$ connects $u_i$ to $v_i$ for all $i$,
    \item $Q_i$ is in $H_i$ for all $i$, and
    \item for all $i,j$, if $\{i,j\} \in S$ then $Q_i$ and $Q_j$ are vertex-disjoint. 
\end{itemize} 
Note that we do not require the paths to be shortest paths; in fact the graph $H$ is unweighted. Schrijver~\cite{S15} solved the partially vertex-disjoint paths problem for fixed $h$ in polynomial time. He does not state the running time of the algorithm, but it appears to be $(\text{poly}(|V(H)|))^{h^2}$.

PNEPP is the same as PVPP except that if $\{i,j\} \in S$, then we require the directed paths $Q_i$ and $Q_j$ to be non-crossing edge-disjoint paths instead of vertex-disjoint paths. (Recall that by ``edge-disjont'' we mean that if $Q_i$ uses $e$ then $Q_i$ can use neither $e$ nor $rev(e)$.)

Now we describe the reduction from PNEPP to PVPP.
Suppose we are given an instance $H$ of PNEPP with terminal pairs $(u_1,v_1), \dots, (u_h, v_h)$ and a set $S$ of pairs of indices of terminals, and subgraphs $H_1, \dots, H_h$. We construct an instance $H'$ of PVPP by replacing each non-terminal vertex $v$ with an $2h \times 2h$ grid $g_v$ of bidirected edges, where $n = |V(G)|$. (The grid can be made smaller, say $p_v \times p_v$ where $p_v = \max\{k, \deg(v)\}$, but this suffices for our purposes.) Every arc that was incident to vertex $v$ in $H$ is instead incident to a vertex on the boundary of $g_v$; furthermore, we can make it so that no two arcs in $H$ share endpoints in $H'$. See Figure~\ref{F:grid-reduction}. The subgraphs $H_1, \dots, H_h$ and the terminals $s_1, t_1, \dots, s_h, t_h$ are the same in $H'$ and $H$. To show that this reduction is correct we need to prove the following lemma:

\begin{figure}
    \centering
    	\begin{tabular}{crcrcr}
    	\includegraphics[scale = 0.35]{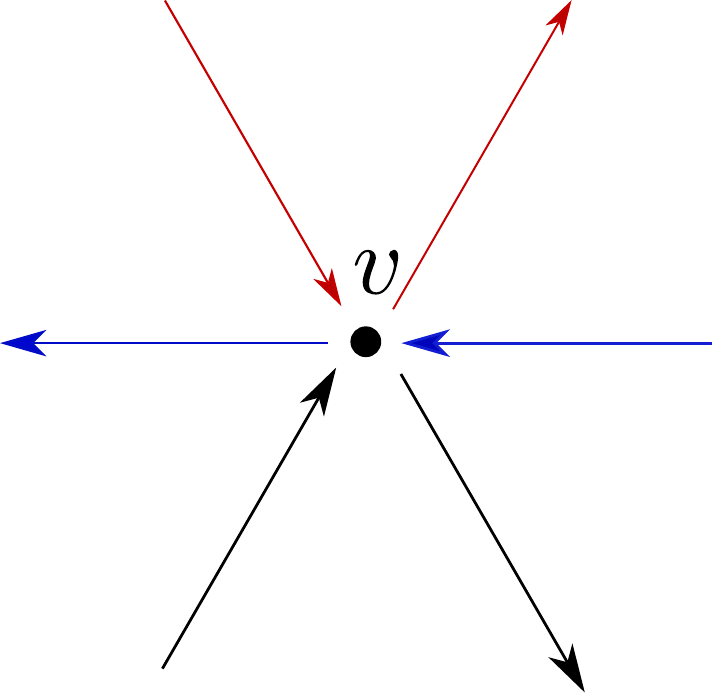} & \hspace{-0.25in} (a)
    	&
    	\includegraphics[scale = 0.35]{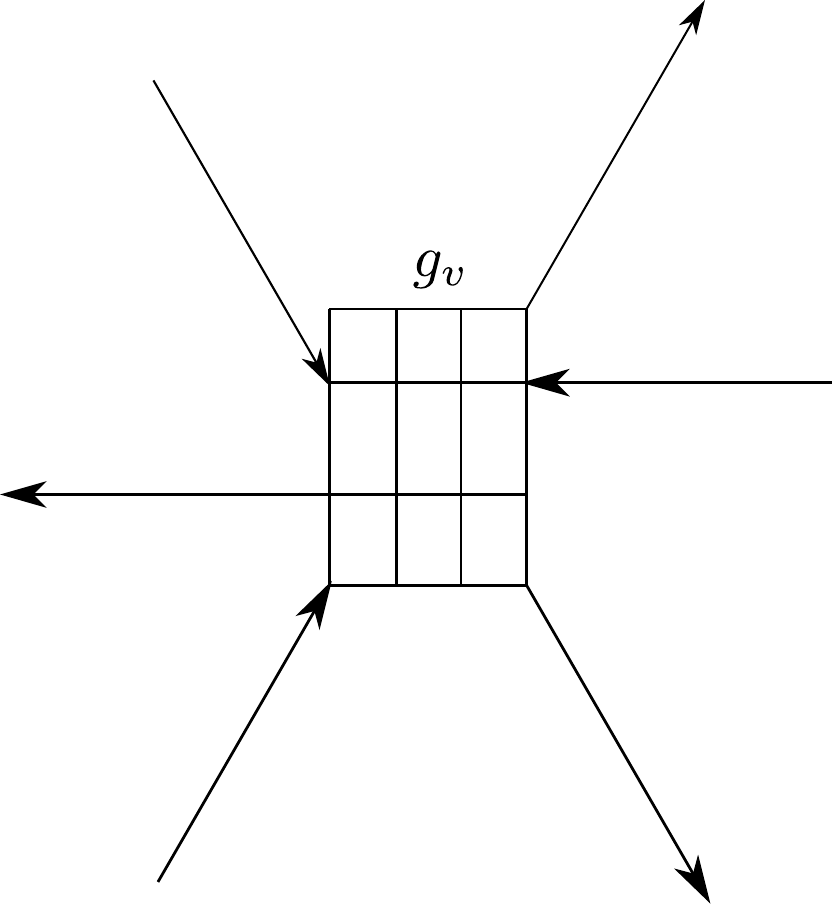} & \hspace{-0.25in} (b)&
    	\includegraphics[scale = 0.35]{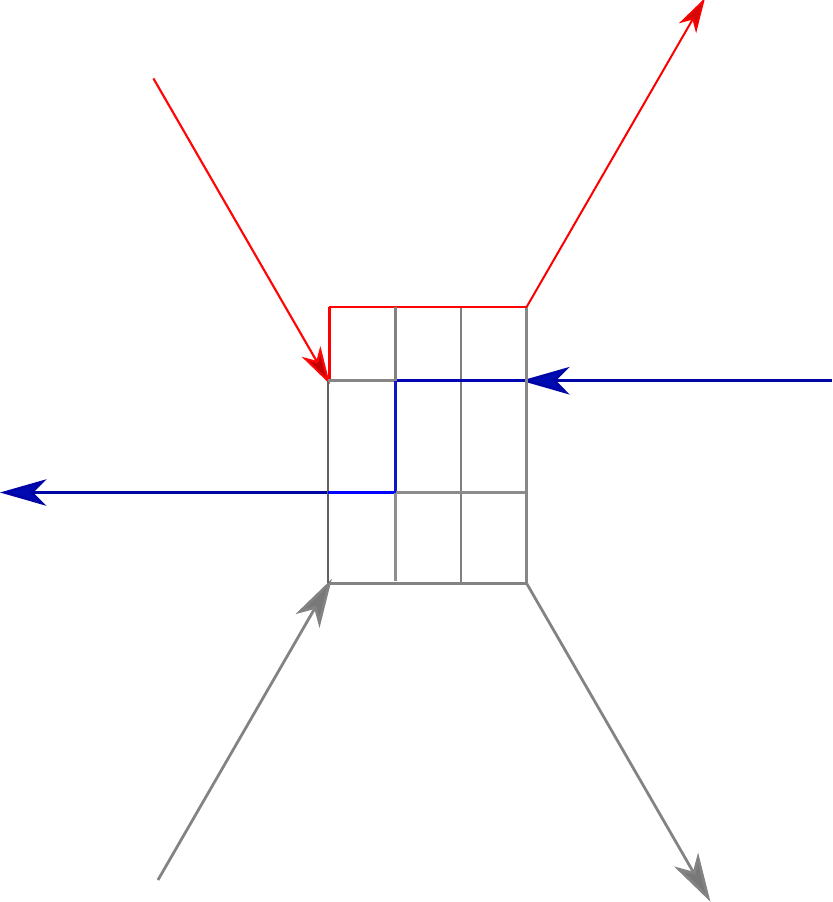} & \hspace{-0.25in} (c)
    \end{tabular}
    \caption{(a) a red and blue path going through a vertex $v$ in $G$ (b) corresponding grid $g_v$ in $H$ with $k = 2$. (c) routing the red and blue paths through $g_v$, in the proof of Lemma~\ref{L:grid-correctness}}
    \label{F:grid-reduction}
    \end{figure}

\begin{lemma}\label{L:grid-correctness}
    The following two statements are equivalent:
    \begin{enumerate}
        \item In $G$, there exist paths $P_1, \dots, P_h$ such that $P_i$ connects $u_i$ to $v_i$, $P_i$ is in $H_i$ for all $i$, and if $\{i,j\} \in S$ then $P_i$ and $P_j$ are non-crossing and edge-disjoint.
        \item In $H$, there exist paths $Q_1, \dots, Q_h$ such that $Q_i$ connects $u_i$ to $v_i$, $Q_i$ is in $H_i$ for all $i$, and if $\{i,j\} \in S$ then $Q_i$ and $Q_j$ are vertex-disjoint.
    \end{enumerate}
\end{lemma}
\begin{proof}
    \underline{$\Rightarrow:$} Suppose that non-crossing partially edge-disjoint paths $P_1, \dots, P_h$ exist in $G$. We construct paths $Q_1, \dots, Q_h$ as follows. For any arc $e$ in $P_i$, we add $e$ to $Q_i$. This defines the portions of the paths $Q_1, \dots, Q_h$ outside the grids $g_v$; these portions are vertex-disjoint because by construction the endpoints of $G$ are all distinct. 
    
    To find the portions of $Q_1, \dots, Q_h$ inside a single grid $g_v$, we need to solve the following problem. Suppose $k'$ of the paths $P_1, \dots, P_h$ went through $v$ in $G$. Re-index the paths such that $P_1, \dots, P_{k'}$ go through $v$ and $P_{k'+1}, \dots, P_h$ do not. 
    We are given a subgraph $g$ of the $n \times n$ bidirected grid with $k'$ pairs of non-crossing terminals $(w_1, x_1), \dots, (w_{k'}, x_{k'})$ on the boundary of $g$, and we want to find pairwise vertex-disjoint paths in $g$ such that the $i$-th path $\pi_i$ connects $w_i$ to $x_i$. To solve this problem, we route the paths one by one as follows. List the terminals $w_1, x_1, \dots, w_{k'}, x_{k'}$ in cyclic order around the outer face of $g$; there must be some $i$ such that the two vertices $w_i$ and $x_i$ appear consecutively in this list. Terminals $w_i$ and $x_i$ split the boundary of $g$ into two segments; we let $\pi_i$ be the portion of the boundary that does not contain any other terminals. Remove the vertices of $\pi_i$ from $g$ and recursively compute the other paths $\pi_1, \dots, \pi_{i-1}, \pi_{i+1}, \dots, \pi_{h'}$.
    
    Routing $\pi_i$ is possible as long as $g$ is connected. Each time we recurse, the outerplanarity index of the $g$ goes down by at most 1. Initially, $g$ is the $2h \times 2h$ grid, so the outerplanarity index of $g$ starts at $h \geq k'$.
    Thus our recursive algorithm is able to connect all the pairs $(x_1, w_1), \dots, (x_{h'}, w_{h'})$.
    
    \underline{$\Leftarrow:$} Suppose partially vertex-disjoint paths $Q_1, \dots, Q_h$ exist in $H$. Trivially, the paths $Q_1, \dots, Q_h$ are non-crossing partially edge-disjoint too. Each path $P_i$ can be defined to be the ``projection'' of $Q_i$ into $G$ in the obvious way: an arc $e$ of $G$ is in $P_i$ if and only if $e$ was in the original path $Q_i$. 
    
    The paths $P_1, \dots, P_k$ are non-crossing because the paths $Q_1, \dots, Q_k$ are non-crossing. We now show that the paths $P_1, \dots, P_k$ are pairwise edge-disjoint. Suppose for the sake of argument that $P_i$ and $P_j$ share an arc $uv$. Arc $uv$ is in the original graph $G$, so it must connect the grid $g_u$ to the grid $g_v$. Since there is only one edge in $H$ that connects $g_u$ to $g_v$, this means that $Q_i$ and $Q_j$ both use this arc, and so are not vertex-disjoint.
\end{proof}
The reduction clearly runs in polynomial time.

\section{Structure}
Let $a, b, c$, and  $d$ be four vertices on the outer face of $G$. 
Let $P$ be a directed walk from $a$ to $b$ and let $Q$ be a directed walk from $c$ to $d$.
Walks $P$ and $Q$ are {\em opposite} if the cyclic order of their four endpoints around $\partial G$ is $a, b, c, d$. $P$ and $Q$ are {\em parallel} if the order is $a, b, d, c$, and we denote this by $P \sim Q$. We define each path to be parallel to itself. Note that if $P$ is parallel to $Q$, then $Q$ is parallel to $P$. In addition, if $P$ is parallel to $Q$ and $Q$ is parallel to a directed walk $R$, then $P$ is parallel to $R$. Thus $\sim$ is an equivalence relation.
We have the following two lemmas.

\begin{lemma}\label{L:opposite}
Suppose $G$ has positive edge weights, and suppose $P$ and $Q$ are opposite non-conflicting shortest paths. If a vertex $x$ precedes a vertex $y$ on $P$, then $x$ does not precede $y$ in $Q$. In particular, $P$ and $Q$ are edge-disjoint. 
\end{lemma}
\begin{proof}
    Suppose for the sake of argument that $P$ and $Q$ are opposite non-conflicting shortest paths, and vertex $x$ precedes vertex $y$ on both $P$ and $Q$.
    By the Jordan curve theorem, there exists a vertex $z$ on $P \cap Q$ such that either $z$ precedes $x$ on $Q$ and $y$ precedes $z$ on $P$, or $y$ precedes $z$ on $Q$ and $z$ precedes $x$ on $P$. Suppose the first case holds. See Figure~\ref{F:structure}a. 
    Since $P$ and $Q$ are shortest paths, we have 
    \begin{align*}
        &\ell(P[z,x]) = \ell(Q[x,z]) = \ell(Q[x,y]) + \ell(Q[y,z]) \text{ and}\\
        &\ell(P[z,x]) + \ell(P[x,y]) = \ell(P[z,y]) = \ell(Q[y,z]).
    \end{align*} 
    This is impossible because $\ell(P[x,y]) = \ell(Q[x,y]) > 0$.
    
    Now suppose the second case holds. See Figure~\ref{F:structure}b. Similar to the previous case, we have
    \begin{align*}
        &\ell(P[y,z]) = \ell(Q[z,y]) = \ell(Q[z,x]) + \ell(Q[x,y]) \text{ and}\\
        &\ell(P[x,y] + \ell(P[y,z]) = \ell(P[x,z]) = \ell(Q[z,x]),
    \end{align*} 
    which is impossible because $\ell(Q[x,y]) = \ell(P[x,y]) > 0$.

\end{proof}

    \begin{figure}
    \centering
    	\begin{tabular}{cr@{\qquad}cr@{\qquad}cr}
    	\includegraphics[scale = 0.35]{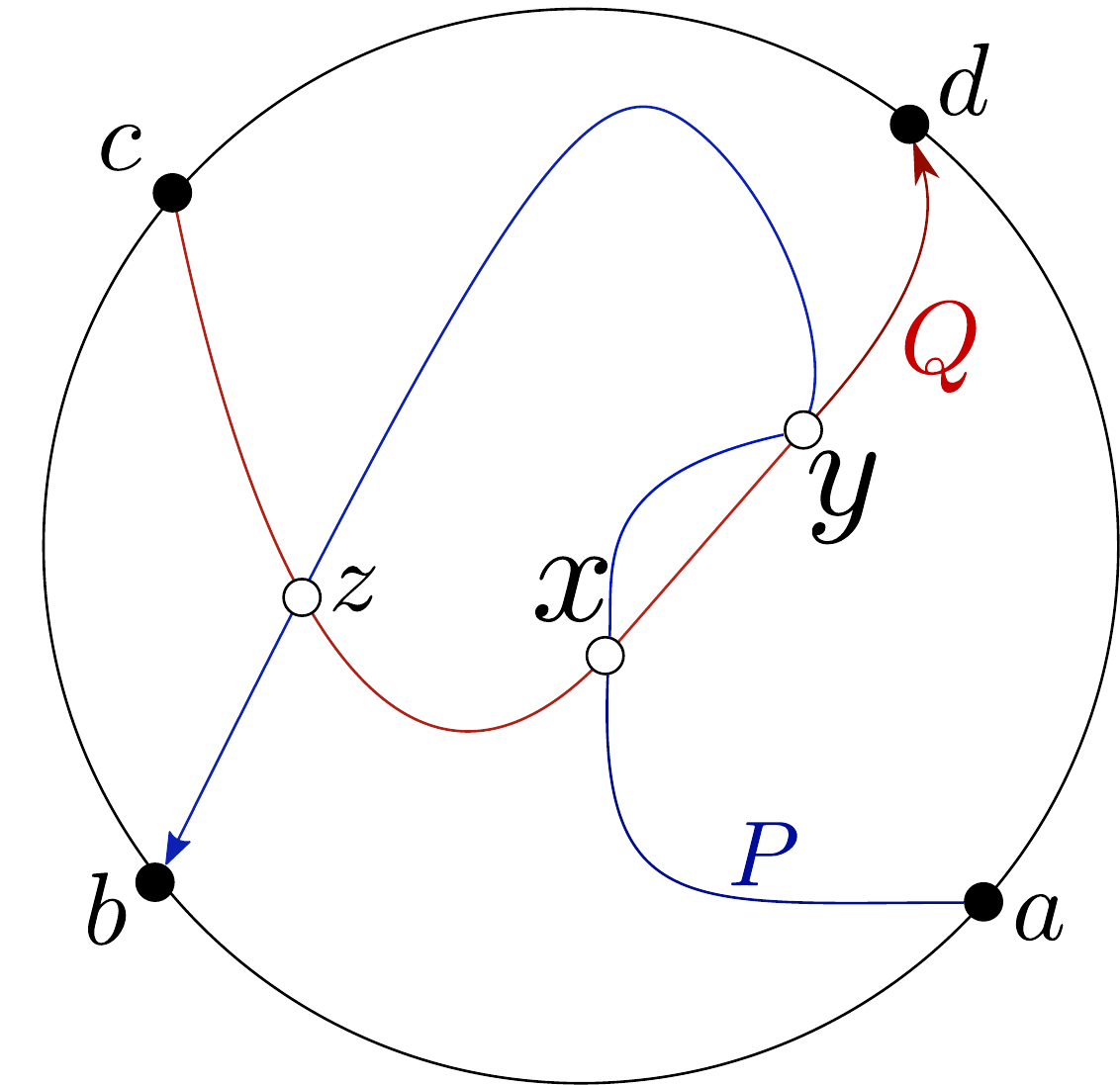} & \hspace{-0.25in} (a)
    	&
    	\includegraphics[scale = 0.35]{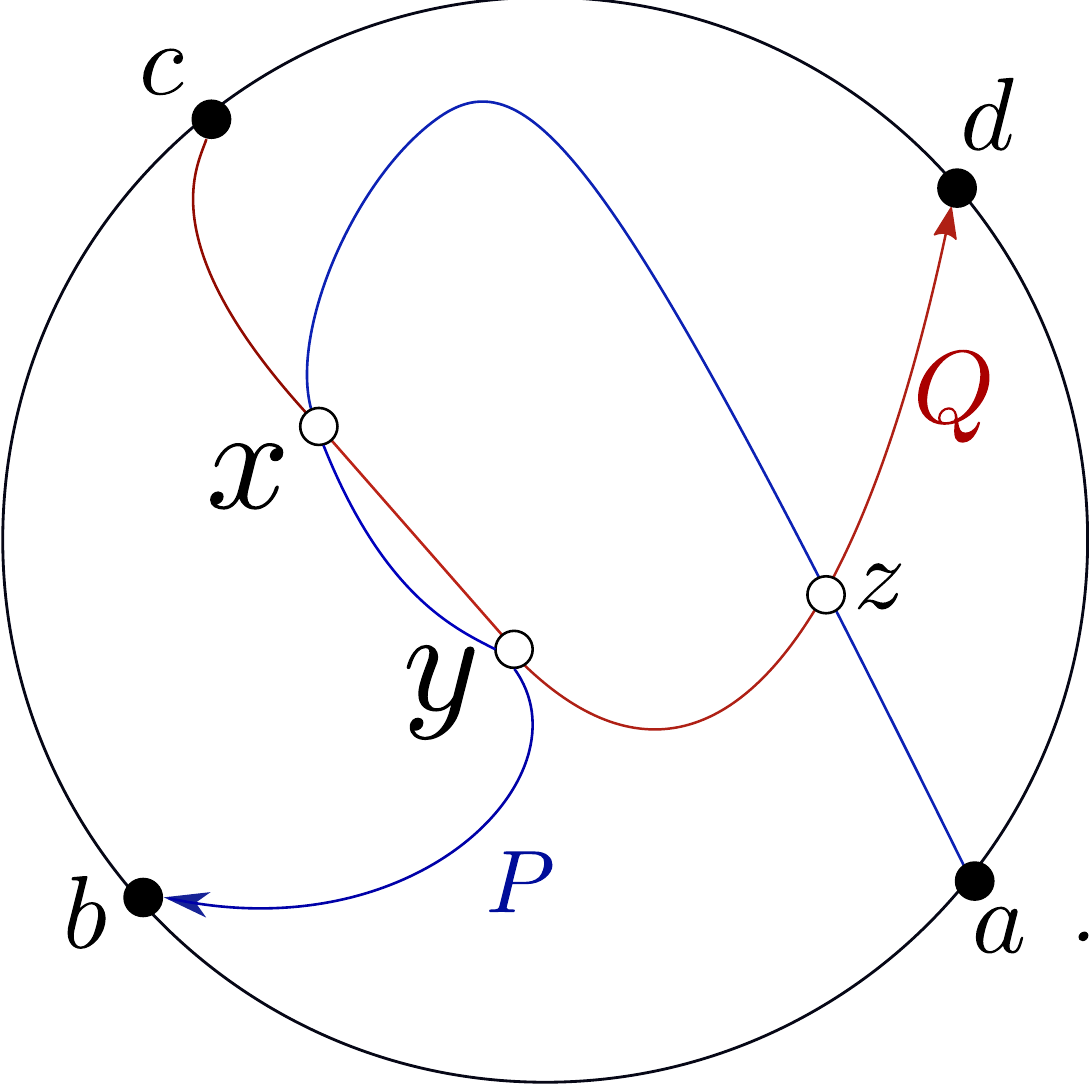} & \hspace{-0.25in} (b)
    \end{tabular}
    \caption{Impossible configurations in the proof of Lemma~\ref{L:opposite}. The blue path is $P$ and the red path is $Q$ (a) $z \prec_Q x$ and $y \prec_P z$ (b) $z \prec_P x$ and $y \prec_Q z$}
    \label{F:structure}
    \end{figure}

\begin{lemma}\label{L:codirectional}
    Suppose $G$ has positive edge weights, and suppose $P$ and $Q$ are parallel shortest paths. If a vertex $x$ precedes a vertex $y$ in $P$, then $y$ does not precede $x$ in $Q$. In particular, $P$ and $Q$ do not conflict. 
\end{lemma}
\begin{proof}
    This is just Lemma~\ref{L:opposite} with $Q$ replaced by $rev(Q)$.
\end{proof}

This lemma immediately suggests an algorithm for a special case of the ideal orientation problem. Suppose $G$ is an instance where the terminals all appear on the outer face in clockwise order $s_1, \dots, s_k, t_k, \dots, t_1$.
Lemma~\ref{L:codirectional} implies that the shortest paths from $s_i$ to $t_i$ are non-conflicting, so we just need to find a shortest path from $s_i$ to $t_i$ for all $i$. Steiger~\cite{S17} showed how to find a representation of these paths in $O(n \log \log k)$ time.

The following two lemmas are trivial when shortest paths are unique. In the ideal orientation problem, we cannot assume that shortest paths are unique because then the problem becomes trivial: just find the (unique) shortest paths and check if they conflict.

\begin{lemma}\label{L:uncross-codirectional}
    Let $G$ be any planar instance of the ideal orientation problem with terminal pairs $(s_1, t_1)$, $\dots$, $(s_k, t_k)$. If a solution $\mathcal{P}$ exists, then a solution $\mathcal{P}'$ exists in which for every pair of parallel paths $P_i$ and $P_j$ in $\mathcal{P}$, 
 $P_i$ and $P_j$ are non-crossing. 
\end{lemma}
\begin{proof}
    This was proved by Liang and Lu~\cite{LL17}. For details see Appendix~\ref{A:LL17}.
\end{proof}

\begin{lemma}\label{L:connected-intersection}
    Let $G$ be any planar instance of the ideal orientation problem with terminal pairs $(s_1, t_1)$, $\dots$, $(s_k, t_k)$ and such that parallel paths do not cross. If a solution $\mathcal{P}$ exists, then a solution $\mathcal{P}'$ exists in which for every pair of parallel paths $P_i$ and $P_j$ in $\mathcal{P}$, $P_i \cap P_j$ is connected.
    Furthermore, we can enforce the second property such that the number of crossings between paths of $\mathcal{P}'$ is no more than the number of crossings between paths of $\mathcal{P}$. 
\end{lemma}
\begin{proof}

	Suppose we have a solution $\mathcal{P}$ in which parallel paths do not cross. We show how to make the second property holds while ensuring that parallel paths remain non-crossing and the total number of crossings between opposite paths does not increase.

	Re-index the terminal pairs such that $\mathcal{P} = \{P_1, \dots, P_h\}$ are pairwise non-crossing parallel paths and form an equivalence class.
 	It is straightforward to verify that the clockwise order of the terminals is then $s_1, \dots, s_h, t_h, \dots t_1$.

    Suppose $P_i \cap P_{i+1}$ consists of at least two paths, where $i < h$. There exist vertices $x$ and $y$ on $P_i \cap P_{i+1}$ such that $P_i[x,y]$ and $P_{i+1}[x,y]$ intersect only at $x$ and $y$. Note that no paths in $\mathcal{P}$ enter the interior of $P_i[x,y] \cup P_{i+1}[x,y]$.
    There are two cases
    \begin{enumerate}
        \item Suppose $P_i[x,y]$ crosses at least as many opposite paths than $P_{i+1}[x,y]$. Then we can perform an exchange as follows. Suppose $P_j, \dots, P_i$ are the paths that contain $P_i[x,y]$. Define $P_p' = P_p[s_p,x] \circ P_{i+1}[x,y] \circ P_p[y,t_p]$ for all $p \in [j,i]$, and let $\mathcal{P}' = \mathcal{P} \setminus \{P_j, \dots, P_i\} \cup \{P_j', \dots, P_i'\}$. Since $P_{i+1}[x,y]$ is a shortest path that does not conflict with any path in $\mathcal{P}$, $P_p'$ is a shortest path from $s_p$ to $t_p$ that does not conflict with any path in $\mathcal{P}'$ for all $p \in [j,i]$. Since no parallel paths enter the interior of $P_[x,y] \cup P_{i+1}[x,y]$, $P_p'$ does not cross any parallel path for $p \in [j,i]$. Finally, the number of opposite paths that $P_p'$ crosses is at most the number of opposite paths that $P_p$ crosses for $p \in [j,i]$, since $P_{i+1}[x,y]$ crosses at most as many opposite paths as $P_{i}[x,y]$ does.
        \item Suppose $P_i[x,y]$ crosses fewer opposite paths than $P_{i+1}[x,y]$. Then we can exchange $P_{i+1}[x,y]$ for $P_{i}[x,y]$. That is, suppose $P_i, \dots, P_j$ are the paths that contain $P_i[x,y]$. Define $P_p' = P_p[s_p,x] \circ P_{i}[x,y] \circ P_p[y,t_p]$ for all $p \in [i,j]$, and let $\mathcal{P}' = \mathcal{P} \setminus \{P_i, \dots, P_j\} \cup \{P_i', \dots, P_j'\}$. Analogous to the previous case, one can show that the number of crossings does not increase and that parallel paths still do not cross. Furthermore, the resulting solution is still made up of shortest paths and is thus still a solution.
    \end{enumerate}
    As long as there exist two paths whose intersection is not a single subpath, we can perform the exchange.
    Each time we perform the exchange, the number of regions that the parallel paths split the interior of $G$ into decreases, so eventually we will no longer be able to perform the exchange. At this point, every pair of parallel paths has intersection consisting of only one subpath. Repeat for all other equivalence classes.
\end{proof}

\section{Serial case for ideal orientations}\label{S:serial-ideal}
Recall that an instance of the ideal orientation problem is {\em serial} if the terminals all appear on the outer face in clockwise order $u_1, v_1, \dots, u_k, v_k$, where for each $i \in [k]$ we have $(u_i, v_i) = (s_i, t_i)$ or $(u_i, v_i) = (t_i, s_i)$.
For all $i$, if $(u_i, v_i) = (s_i, t_i)$, then we say that $(s_i, t_i)$ and any path from $s_i$ to $t_i$ are {\em clockwise}; otherwise $(s_i, t_i)$ and any path from $s_i$ to $t_i$ are {\em counterclockwise}.
Note that a clockwise and a counterclockwise path are parallel, while two clockwise paths (or two counterclockwise paths) are opposite.
In this section, we describe an algorithm that solves serial instances of the ideal orientation problem in $O(n^2)$ time even when $k$ is part of the input. First we prove the following lemmas:

\subsection{Envelopes}
Suppose $G$ is a serial instance with terminal pairs $(s_1, t_1), \dots, (s_k, t_k)$. Suppose we have a set $\Pi$ of arbitrary directed paths $\pi_1, \dots, \pi_k$ such that $\pi_i$ connects $s_i$ to $t_i$ and no path touches $\partial G$; the paths may intersect arbitrarily. Let $C_i$ be the portion of $\partial G$ that connects $s_i$ to $t_i$ without containing any other terminals.
The paths in $\pi_1, \dots, \pi_k$ divide the interior of $G$ into connected regions.  Let $R_i$ be the unique region with $C_i$ on its boundary. Finally, we define $L(i, \Pi)$ to be the directed path from $s_i$ to $t_i$ whose set of edges is $\partial R_i \setminus C_i$. Intuitively, $L(i,\Pi)$ is the ``lower envelope'' of $\pi_1, \dots, \pi_k$ if we draw $G$ such that $s_i$ and $t_i$ are on the bottom. See Figure~\ref{F:envelopes}

\begin{figure}
    \centering
    \begin{tabular}{cr@{\qquad}cr@{\qquad}cr}
    	\includegraphics[scale = 0.25]{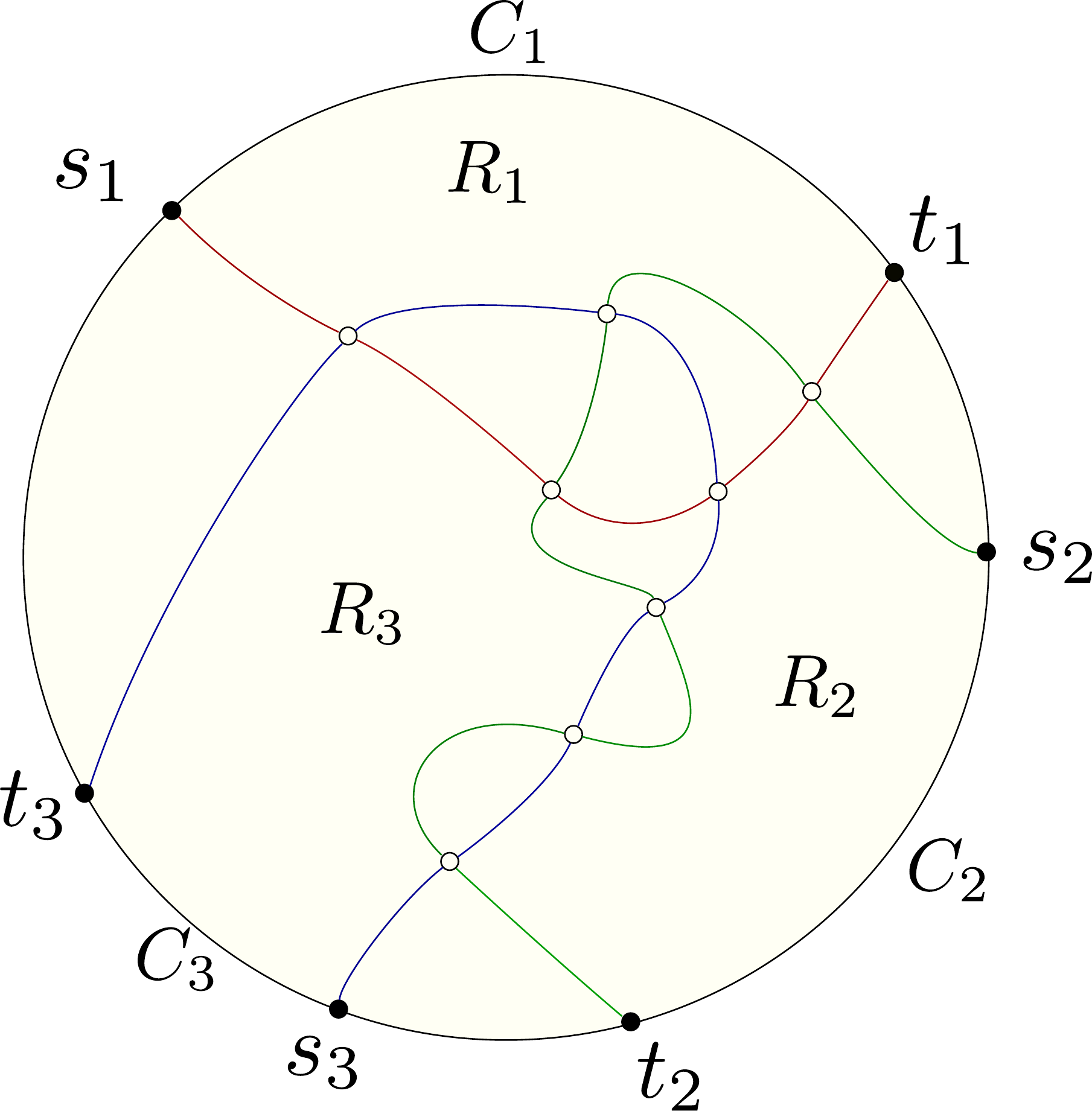} & \hspace{-0.25in} (a)
    	&
    	\includegraphics[scale = 0.25]{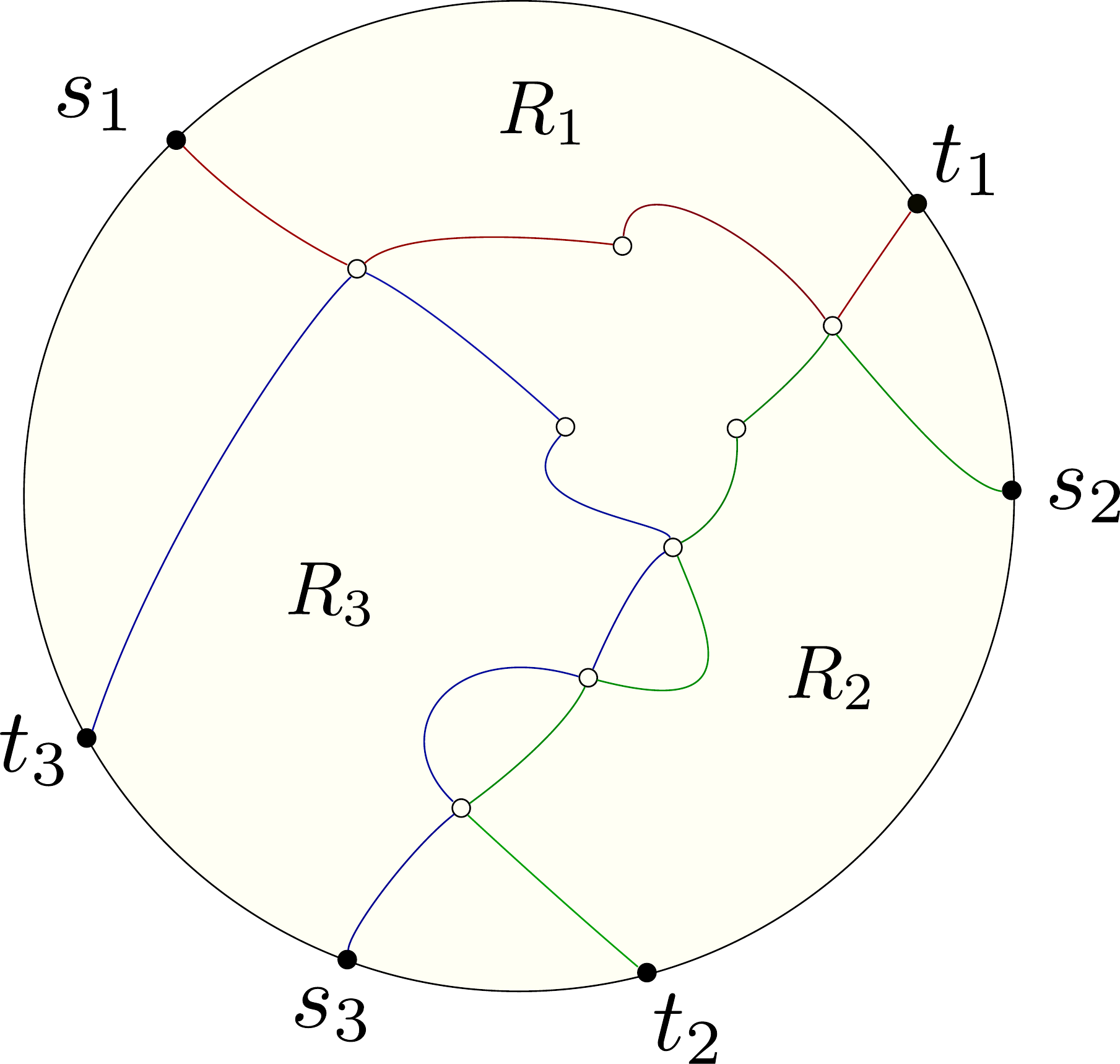} & \hspace{-0.25in} (b)
    \end{tabular}
    \caption{All paths are directed from $s_i$ to $t_i$. (a) We have $\Pi = \{\pi_1, \pi_2, \pi_3\}$, where the red path is $\pi_1$, green path is $\pi_2$, and blue path is $\pi_3$ (b) The red path is $L(1, \Pi)$, green path is $L(2, \Pi)$, and blue path is $L(3, \Pi)$}
    \label{F:envelopes}
    \end{figure}

For any arc $e$ in $L_i$, either $e$ or $rev(e)$ must be an arc in one of $\pi_1, \dots, \pi_k$. Thus $L(i,\Pi)$ does not contain any edges on $\partial G$. Also, the walks $L(1, \Pi), \dots, L(k, \Pi)$ are pairwise non-crossing.

\begin{lemma}\label{L:envelope-nonconflicting}
    Suppose $G$ is serial and $\Pi = \{\pi_1, \dots, \pi_k\}$ is a set of paths such that $\pi_i$ connects terminal $s_i$ to terminal $t_i$ for all $i$. 
    If $\pi_1, \dots, \pi_k$ are pairwise non-conflicting, then $L(1,\Pi), \dots, L(k,\Pi)$ are also pairwise non-conflicting.
\end{lemma}. 
\begin{proof}
    We prove the contrapositive. Suppose $L(i,\Pi)$ and $L(j,\Pi)$ conflict at an edge $e$. Then the regions $R_i$ and $R_j$ touch each other at $e$. Since $\pi_i$ and $\pi_j$ do not use any boundary arcs, the Jordan Curve Theorem implies that $\pi_i$ and $\pi_j$ also conflict at $e$.
\end{proof}

\begin{lemma}\label{L:envelope-length}
    Suppose $G$ is serial and $\Pi = \{\pi_1, \dots, \pi_k\}$ is a set of pairwise non-conflicting paths such that $\pi_i$ connects terminal $s_i$ to terminal $t_i$ for all $i$. 
    Then we have 
    \[
        \sum_{i = 1}^k \ell(\pi_i) \geq \sum_{i=1}^k \ell(L(i, \Pi)).
    \]
    In particular, if $\pi_1, \dots, \pi_k$ are shortest paths, then so are $L(1,\Pi), \dots, L(k,\Pi)$.
\end{lemma} 
\begin{proof}
    Because the walks $L(1, \Pi), \dots, L(k, \Pi)$ are pairwise non-crossing, the Jordan Curve Theorem implies that each arc $e$ can only be used by at most one of the walks $L(1, \Pi), \dots, L(k, \Pi)$. By Lemma~\ref{L:envelope-nonconflicting}, arc $rev(e)$ can only be used by one of the walks $L(1, \Pi), \dots, L(k, \Pi)$ if $e$ is not used by any of those walks. It follows that each edge of $G$ is used by at most one of the walks $L(1, \Pi), \dots, L(k, \Pi)$. In addition, every edge used by one of $L(1, \Pi), \dots, L(k, \Pi)$ must be used by at least one of $\pi_1, \dots, \pi_k$. The lemma follows.
\end{proof}

\subsection{Algorithm}
\begin{lemma}\label{L:no-crossings}
    Let $G$ be a serial instance of the ideal orientation problem with terminal pairs $(s_1, t_1), \dots, (s_k, t_k)$. If a solution exists, then a solution exists in which the paths $P_1, \dots, P_k$ are pairwise non-crossing.
\end{lemma}
\begin{proof}
    This is straightforward using envelopes.
    Let $\mathcal{P} = \{P_1, \dots, P_k\}$ be a solution to the serial instance $G$ of the ideal orientation problem, where $P_i$ connects $s_i$ to $t_i$.
    Then the walks $L(1, \mathcal{P}), \dots, L(k, \mathcal{P})$ are pairwise non-crossing. By Lemma~\ref{L:envelope-nonconflicting}, the walks are pairwise non-conflicting, and by Lemma~\ref{L:envelope-length} they are shortest paths, so they constitute a solution.

\end{proof}
A path from $s_i$ to $t_i$ is the {\em outermost} shortest path from $s_i$ to $t_i$ if it is outside all other shortest paths from $s_i$ to $t_i$. 
The following lemma states that finding outermost shortest paths is sufficient:
\begin{lemma}
	Let $G$ be a serial instance of the ideal orientation problem with terminal pairs $(s_1, t_1), \dots, (s_k, t_k)$.
	If a solution $\mathcal{P} = \{P_1, \dots, P_k\}$ exists, then a solution exists in which $P_i$ is the outermost shortest path from $s_i$ to $t_i$ for all $i$.
\end{lemma}
\begin{proof}
    Suppose we have a solution $\mathcal{P} = \{P_1, \dots, P_k\}$ where the paths in $\mathcal{P}$ are pairwise non-crossing and some path $P_i$ is not the outermost shortest path from $s_i$ to $t_i$. Then we can exchange $P_i$ for the outermost shortest path. That is, let $P_i'$ be the outermost shortest path from $s_i$ to $t_i$, and let $\mathcal{P}' = \mathcal{P} \setminus \{P_i\} \cup \{P_i'\}$. The new path $P_i'$ is in $R_i$ so it does not cross with any other path in $\mathcal{P}'$; in particular, $P_i'$ does not conflict with any other path in $\mathcal{P}'$, so $\mathcal{P}'$ is a solution. We can keep doing this exchange until we get a solution where every path is the outermost shortest path.
\end{proof}

So the algorithm is to find all outermost shortest paths. If the outermost paths conflict, then there is no solution.
Computing the outermost paths explicitly takes $O(n)$ time per path and thus $O(n^2)$ time~\cite{HKRS97}. Alternately, Klein's algorithm computes an implicit representation of the paths in $O(n \log n)$ time~\cite{K05}. 

\section{General single-face case for fixed $k$ and non-crossing pairs}\label{S:single-face}

In this section we describe an algorithm to solve the ideal orientation problem in planar graphs where all terminals are on a single face, the number of terminal pairs is fixed, and none of the terminal pairs cross. (Recall that two terminal pairs $(s_i, t_i)$ and $(s_j, t_j)$ cross if the four terminals are on a common face and their cyclic order on that face is $s_i, s_j, t_i, t_j$.)
For simplicity, call such instances one-face non-crossing instances. We will first show that the number of crossings between the paths in a solution is a function bounded only by $k$. This allows us to guess all crossing points and then reduce the problem to the partially non-crossing edge-disjoint paths problem (PNEPP). 
The algorithm is inspired by a result of B\'{e}rczi and Kobayashi~\cite{BK17}.

We saw in the previous section that in the serial instances of the ideal orientation problem we can assume that the paths in the solution are non-crossing. This is unfortunately not true for general one-face non-crossing instances. See Figure~\ref{F:eww}a for an example. On the other hand, we can prove that the number of crossings is small when $k$ is small. Specifically, we have the following lemma:

\begin{figure}
    \centering
    \begin{tabular}{cr@{\qquad}cr@{\qquad}cr}
    	\includegraphics[scale = 0.15]{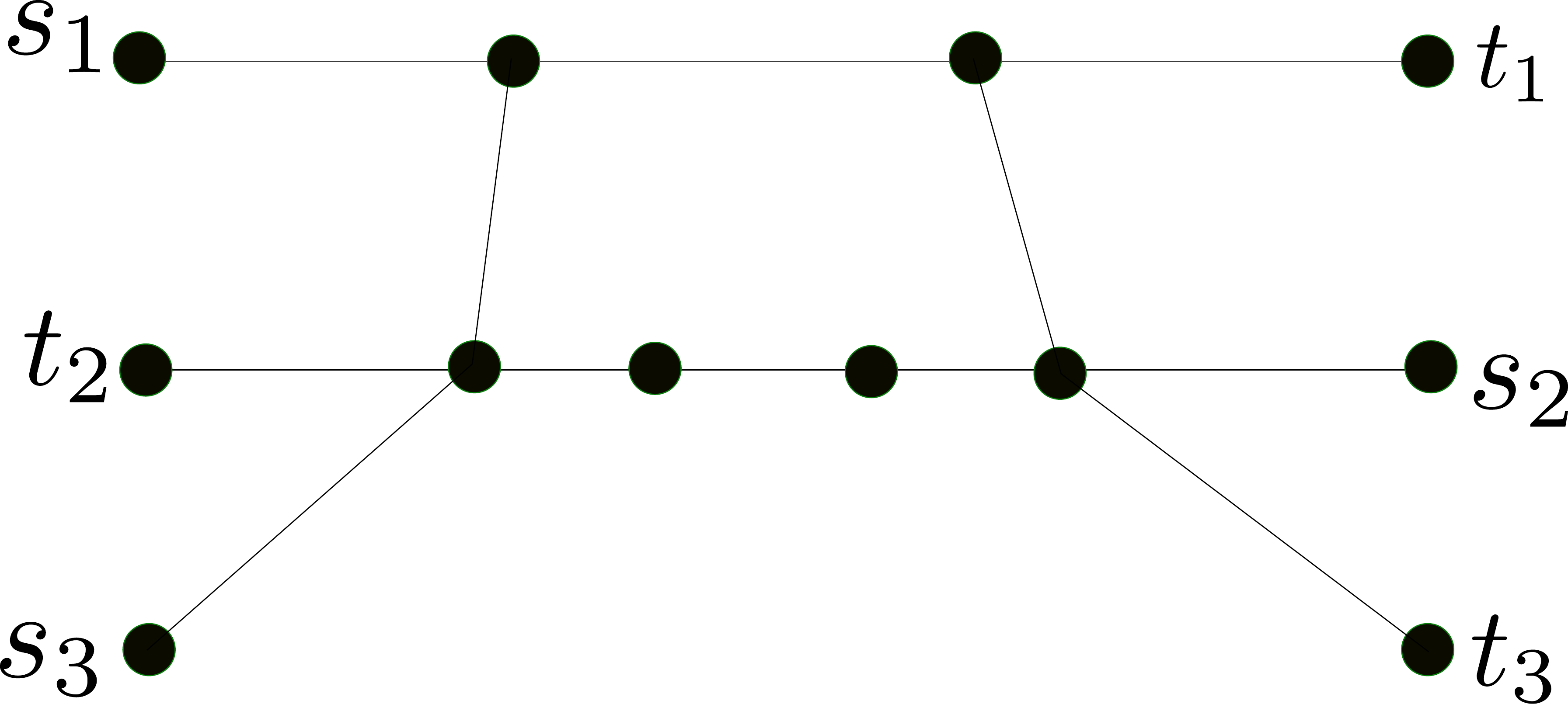} & \hspace{-0.15in} (a) &
    	\includegraphics[scale = 0.15]{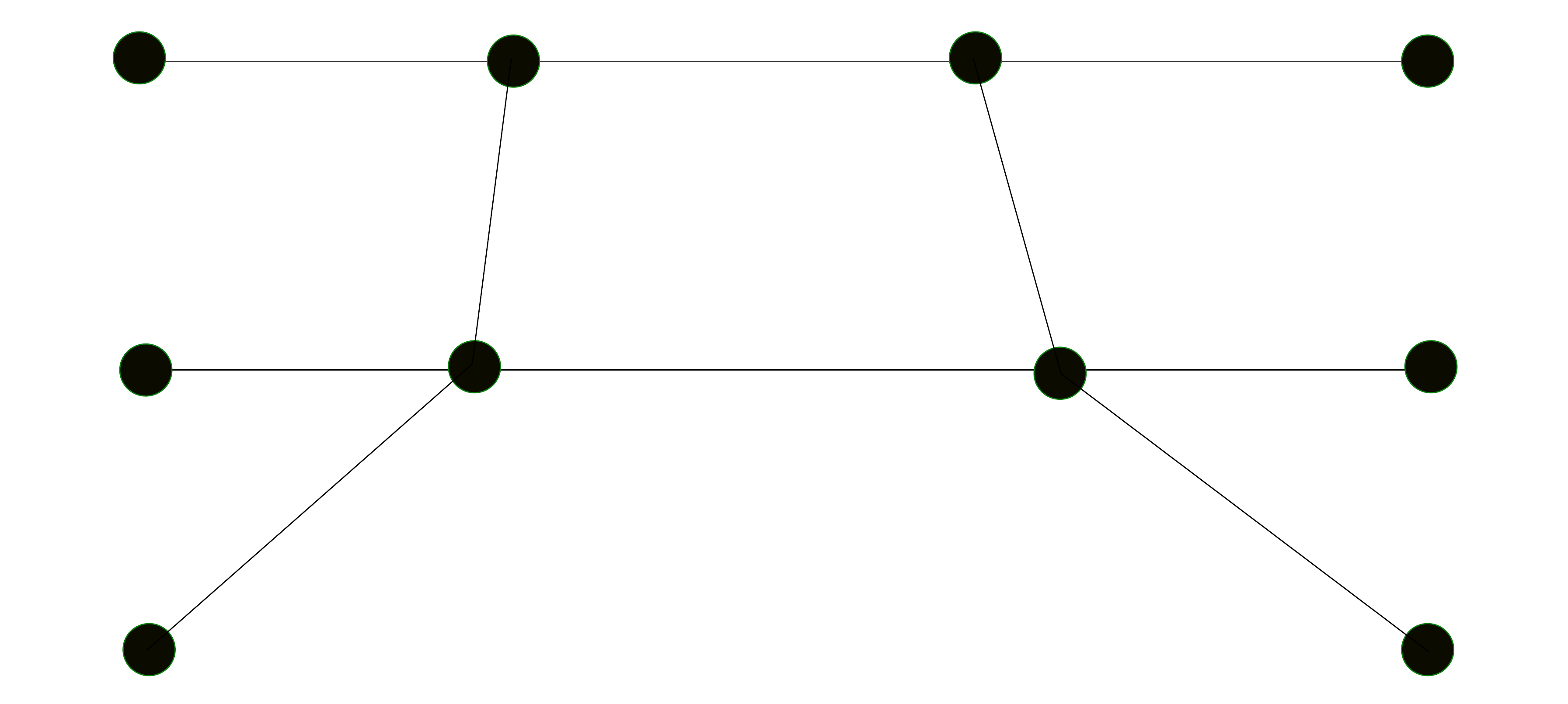} & \hspace{-0.25in} (b)
    	&\\
    	\includegraphics[scale = 0.2]{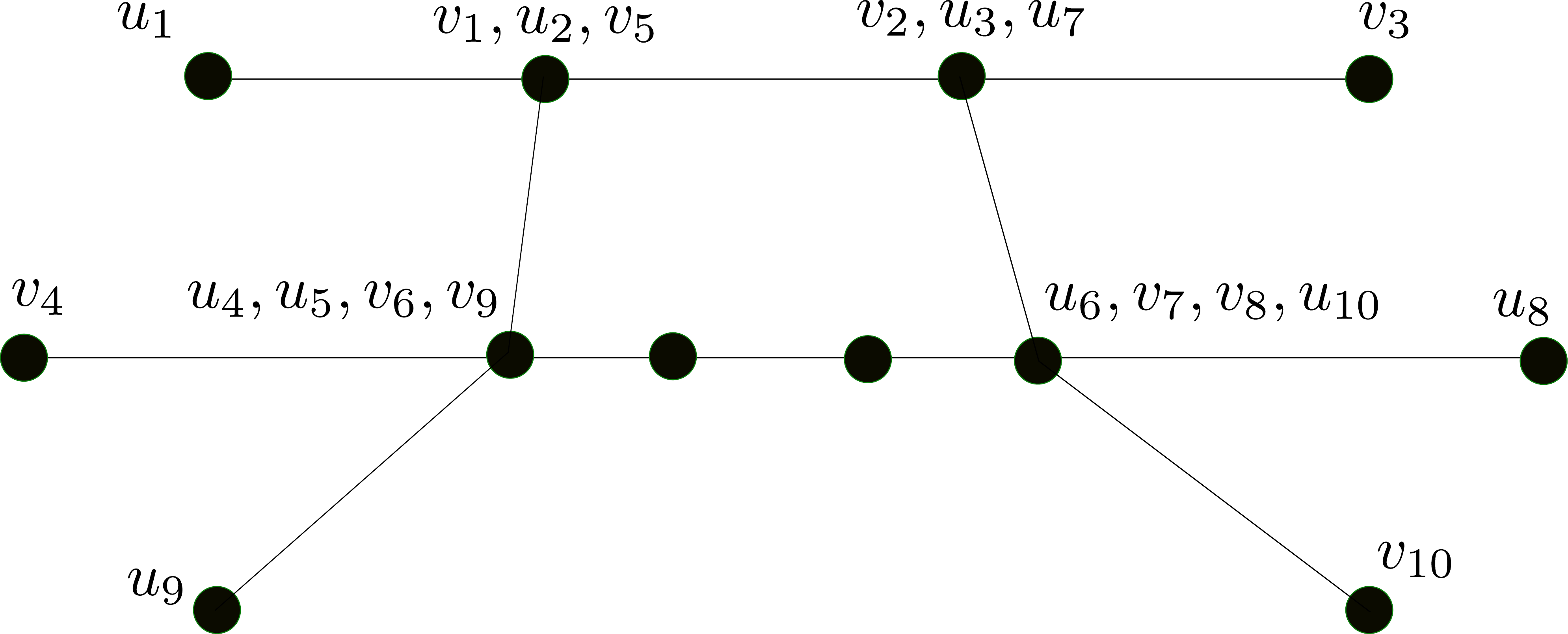} & \hspace{-0.25in} (c)
    \end{tabular}
    	
    \caption{(a) An instance of the ideal orientation problem where the unique solution has the path from $s_2$ to $t_2$ crossing the path from $s_3$ to $t_3$. All edges have unit weight. (b) overlay graph corresponding to the unique solution (c) instance of PNEPP corresponding to the overlay graph. Here $S = \{\{i,j\} | i \in \{1,2,3,5,7,9,10\}, j \in \{4,6,8\}\}$}
    \label{F:eww}
\end{figure}

\begin{lemma}\label{L:few-crossings}
	Suppose $G$ is a one-face non-crossing instance of the ideal orientation problem with terminal pairs $(s_1, t_1)$, $\dots$, $(s_k, t_k)$. If a solution exists, then a solution $\{P_1, \dots, P_k\}$ exists in which for all $i$ and $j$, path $P_i$ crosses $P_j$ a total of $O(k)$ times.
\end{lemma}

We will prove this lemma at the end of the section.
For now, we will just assume the lemma is true and describe the algorithm for one-face non-crossing instances.
Let $G$ be a one-face non-crossing instance of the ideal orientation problem with terminal pairs $(s_1, t_1), \dots, (s_k, t_k)$. 
Suppose $\mathcal{P}  = \{P_1, \dots, P_k$\} is a solution with the fewest crossings. 
Our algorithm first guesses all the crossing points; by Lemma~\ref{L:few-crossings}, there are $O(k^3)$ such points.

Next, inspired by Erickson and Nayyeri \cite{EN11}, we define the overlay graph $H_{\mathcal{P}}$, whose vertices are the crossing points and terminals, and whose edges are the subwalks between consecutive crossing points and terminals. The graph $H_\mathcal{P}$ has a natural embedding. 
Our algorithm guesses the overlay graph. Given the set of $O(k^3)$ crossing points, there are $2^{O(k^6)}$ possible such graphs. See Figure~\ref{F:eww}b.

The $O(k^3)$ crossing points split up $P_1, \dots, P_k$ into pairwise non-crossing directed subpaths. By Lemma~\ref{L:opposite}, subpaths of opposite paths are edge-disjoint (there is no analogous restriction for parallel paths). Furthermore, every directed subpath is a shortest path between its endpoints. Let $\{p_1, \dots, p_\beta\}$ be the set of these directed subpaths.

To compute these subpaths, we construct an instance of PNEPP as follows. The directed graph $H$ is just $G$ where every undirected edge $\{u,v\}$ is replaced with two arcs $uv$ and $vu$. The terminal pairs in $G$ are no longer terminal pairs in $H$. For each subpath $p_i$ in $\mathcal{P}$, we construct a pair of terminals $u_i, v_i$ that are just the endpoints of $p_i$. Thus the constructed terminals will not necessarily be distinct. The set $S$ consists of all pairs $\{i,j\}$ such that $p_i$ and $p_j$ are subpaths of opposite paths in $G$. For all $i \in \{1, \dots, \beta\}$, the subgraph $H_i$ is the union of all shortest paths from $u_i$ to $v_i$ in $H$. Each $H_i$ is a directed acyclic graph.

We then solve the instance $H$ of PNEPP to find subpaths $Q_i$ connecting the $u_i$ to the $v_i$, and we check if the concatenations of the appropriate found subpaths are indeed shortest paths connecting corresponding terminals in $G$. (In Figure~\ref{F:eww}c, we would need to check, for example, that the concatenation of $Q_1$, $Q_2$, and $Q_3$ is indeed a shortest path from $s_1 = u_1$ to $t_1 = v_3$.) If the concatenations are indeed shortest paths in $G$, then they form a solution to the instance $G$ of the ideal orientation problem. Clearly, if we take any solution of $G$, the subpaths formed by the crossing points are non-crossing and non-conflicting.  The following lemma implies that the algorithm is correct.

\begin{lemma}
	Let $G$ be a one-face non-crossing instance of the ideal orientation problem. The following statements are equivalent
	\begin{itemize}
	    \item There exists a solution to $G$ with crossing points $v_1, \dots, v_h$ and overlay graph $H_{\mathcal{P}}$.
	    \item There exist crossing points $v_1, \dots, v_h$, an overlay graph whose vertices are the crossing points and the terminals of $G$, and a set of shortest non-crossing partially edge-disjoint subpaths connecting the crossing points in accordance with the overlay graph, such that the paths formed from concatenating the appropriate subpaths are shortest paths. (Here when we say that two sub-paths are partially edge-disjoint we mean that they are edge-disjoint if they correspond to subpaths of opposite paths in the overlay graph.)
	\end{itemize}
\end{lemma}
\begin{proof}
    $\Rightarrow:$
    Let $\mathcal{P} = \{P_1, \dots, P_k\}$ be a solution to the instance $G$ of the ideal orientation problem.
 	Split the paths in $\mathcal{P}$ into subpaths using the crossing points. The subpaths are non-crossing by construction. We just need to show that the subpaths are partially disjoint. In fact we will show that the paths in $\mathcal{P}$ are partially disjoint. If $P_i$ and $P_j$ are parallel, then there is nothing to prove. If $P_i$ and $P_j$ are opposite, then they are edge-disjoint by Lemma~\ref{L:opposite}. 

    $\Leftarrow:$
    Concatenate the subpaths and assume the concatenations are shortest paths. Since the subpaths are non-crossing, 
    We just need to show that they are non-conflicting. 
    Suppose $P_1, \dots, P_k$ are the resulting paths after concatenation.
    If $P_i$ and $P_j$ are parallel, then by Lemma~\ref{L:codirectional} they are non-conflicting. 
    If $P_i$ and $P_j$ are opposite, then by construction each subpath of $P_i$ is edge-disjoint from each subpath of $P_j$. This means that $P_i$ and $P_j$ are edge-disjoint, so they don't conflict.
\end{proof}

To summarize, we first guess the crossing points, then guess an overlay graph on these crossing points and the original terminals, and finally use the overlay graph to construct and solve an instance of PNEPP. The number of possible sets of crossing points and overlay graphs depends only on $k$, while PNEPP can be solved in polynomial time for fixed $k$ (equivalently, fixed $\beta$) by reducing to PVPP and using Schrijver's algorithm. Furthermore, constructing the instance of PNEPP takes polynomial time. Thus, our algorithm runs in polynomial time for fixed $k$.

\subsection{The crossing bound}\label{SS:crossing-bound}


    In this subsection we prove Lemma~\ref{L:few-crossings}, though some details have been pushed to the appendix.
	Suppose $\mathcal{P} = \{P_1, \dots, P_k\}$ is a solution to a one-face non-crossing instance $G$ of the ideal orientation problem, where $P_i$ connects $s_i$ to $t_i$.
    By Lemmas~\ref{L:uncross-codirectional} and~\ref{L:connected-intersection}we may assume that for every pair of parallel paths in $\mathcal{P}$, the paths do not cross and their intersection consists of at most one subpath; of all such solutions, we may assume without loss of generality that $\mathcal{P}$ is the solution with the fewest crossings.
    It suffices to show that for all $i,j$, path $P_i$ crosses path $P_j$ at most $2k$ times.
    
    Let $h$ be the number of times $P_i$ and $P_j$ cross; we want to show that $h \leq 2k$. 
    Since terminal pairs $(s_i, t_i)$ and $(s_j, t_j)$ are non-crossing, $h$ is even.
    Parallel paths in $\mathcal{P}$ do not cross, so assume that $P_i$ and $P_j$ are opposite. 
    The path $P_i$ divides the interior of $G$ into two regions; let $\rho_i$ be the region containing $s_j$ and $t_j$, and define a path to be {\em above} $P_i$ if it lies in $\rho_i$. Likewise, the path $P_j$ divides the interior of $G$ into two regions; let $\rho_j$ be the region containing $s_i$ and $t_i$, and define a path to be {\em below} $P_j$ if it lies in $\rho_j$.
    Let $x_1, \dots, x_{h}$ be the vertices at which $P_i$ and $P_j$ cross, in order along $P_i$; by Lemma~\ref{L:opposite}, this is exactly the reverse of their order along $P_j$. Split the region $\rho_i \cap \rho_j$ into $h-1$ pairwise internally disjoint {\em bigons}, denoted by $B_1, \dots, B_{h-1}$; the bigon $B_p$ consists of the region bounded by the two subpaths $P_i[x_p,x_{p+1}]$ and $P_j[x_{p+1},x_p]$. 
    Note that under our definition, $P_i[x_p,x_{p+1}]$ and $P_j[x_{p+1},x_p]$ may touch but they may not cross. A bigon $B_p$ is {\em odd} if $p$ is odd and {\em even} if $p$ is even.
    Note that any odd bigon is below $P_i$ and above $P_j$, and any even bigon is below $P_j$ and above $P_i$.
    See Figure~\ref{F:bigons}a.
    
        \begin{figure}
    \centering
    \begin{tabular}{cr@{\qquad}cr@{\qquad}cr}
    	\includegraphics[scale = 0.35]{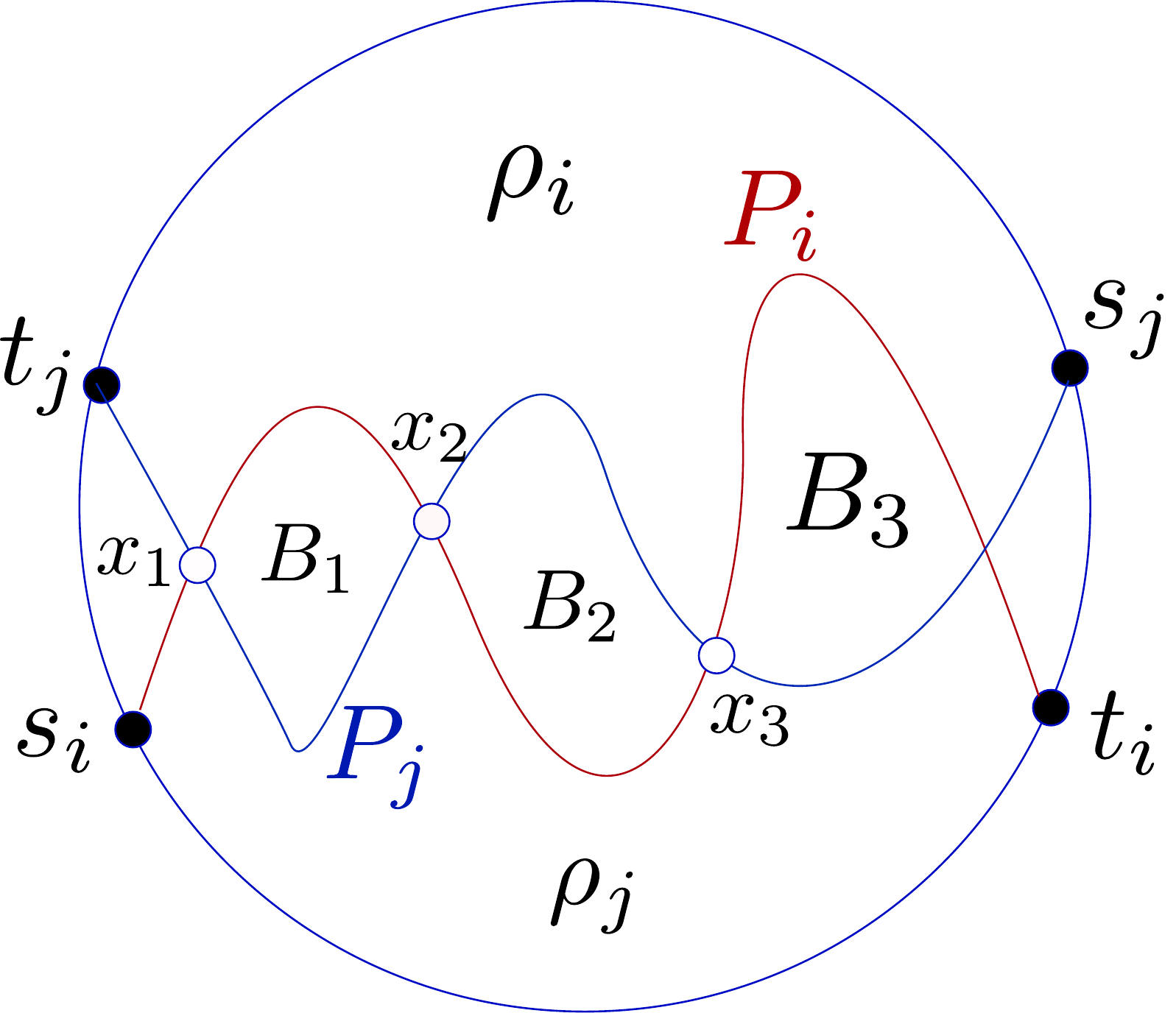} & \hspace{-0.25in} (a)
    	&
    	\includegraphics[scale = 0.35]{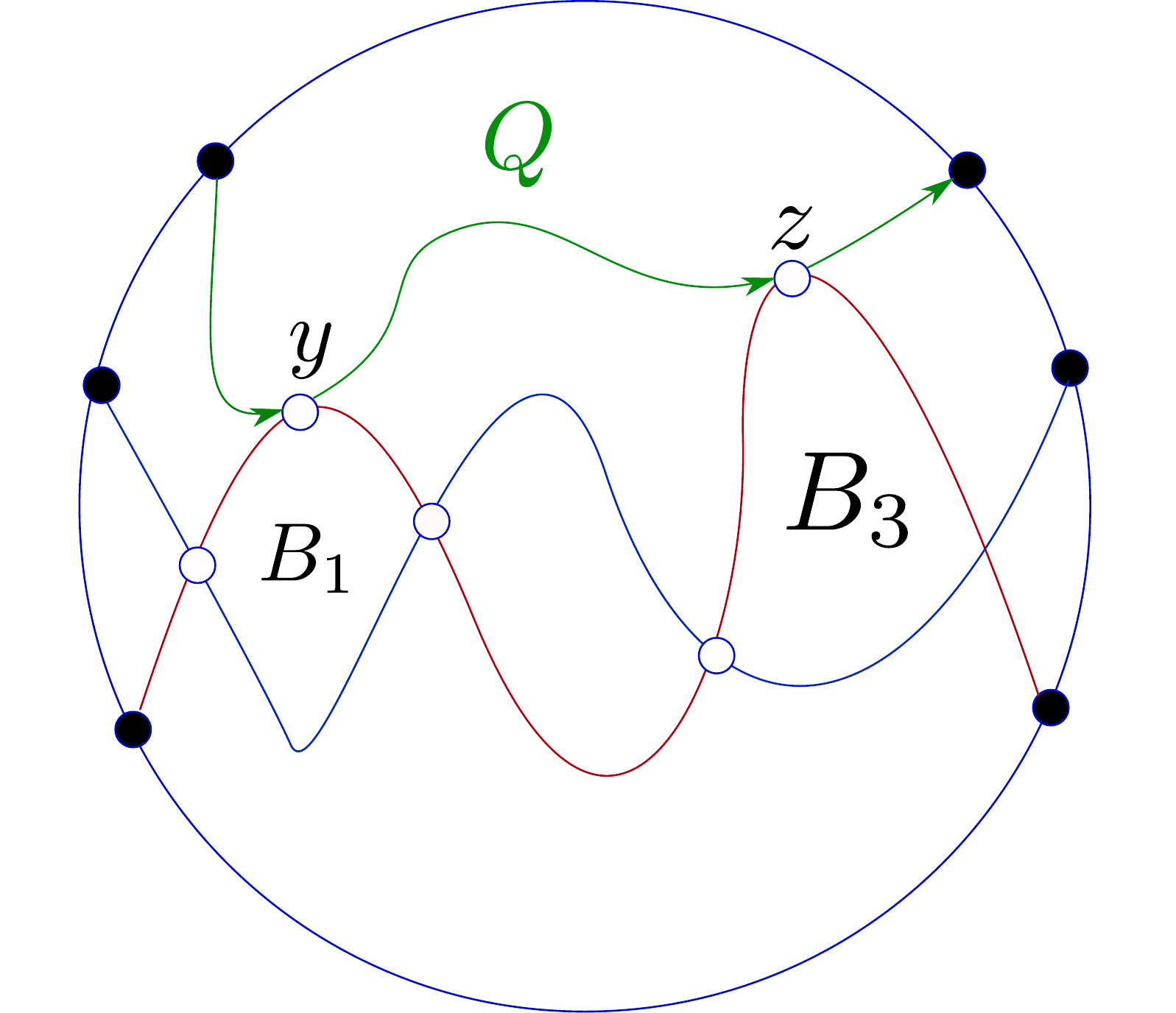} & \hspace{-0.25in} (b)
    \end{tabular}
    	
    \caption{$P_i$ is red and $P_j$ is blue. (a) Three bigons formed by $P_i$ and $P_j$. (b) An impossible configuration in the proof of Lemma~\ref{L:fact-one}. Here $Q$ is the green path, $p=1$, and $q = 3$}
    \label{F:bigons}
    \end{figure}
    
    For any vertex $x$, let $pred_i(x)$ denote the predecessor of $x$ on $P_i$, $succ_i(x)$ denote the successor of $x$ on $P_i$, $pred_j(x)$ denote the predecessor of $x$ on $P_j$, and $succ_j(x)$ denote the successor of $x$ on $P_j$.
    Suppose a path $Q$ is parallel to $P_i$. Path $Q$ and a bigon $B_p$ {\em partially overlap} each other if $Q$ shares edges with $P_i[x_p,x_{p+1}]$ and $Q$ does not contain $P_i[pred_i(x_p),succ_i(x_{p+1})]$.
    Likewise, suppose a path $Q'$ is parallel to $P_j$. Path $Q'$ and a bigon $B_p$ {\em partially overlap} each other if $Q$ shares edges with $P_j[x_{p+1},x_p]$ but $Q$ does not contain $P_j[pred_j(x_{p+1}),succ_j(x_p)]$. 
    
    For the rest of this subsection we say ``overlap'' when we mean ``partially overlap.''
    Lemma~\ref{L:few-crossings} follows if we can prove the following two lemmas:  
    \begin{lemma}\label{L:fact-one}
        Each path in $\mathcal{P}$ overlaps at most two different bigons. 
    \end{lemma}
    \begin{lemma}~\label{L:fact-two}
        Each bigon overlaps some path (different bigons could overlap different paths).
    \end{lemma}
    Lemmas~\ref{L:fact-one} and~\ref{L:fact-two} together imply that there must be at most $2(k-1)$ bigons, and thus at most $2k$ crossings between $P_i$ and $P_j$. 
    
    \begin{proof}[Proof of Lemma~\ref{L:fact-one}]
        There are three cases. For the first case, suppose $Q$ is a path parallel to $P_i$ that overlaps some bigons formed by $P_i$ and $P_j$. We will show that $Q$ overlaps at most two bigons. See Figure~\ref{F:bigons}b. Let $B_p$ be the first bigon along $P_i$ that $Q$ overlaps and let $B_q$ be the last, so that $Q$ contains some vertex $y$ on $P_i[x_p,x_{p+1}]$ and some vertex $z$ on $P_i[x_q, x_{q+1}]$. We have assumed that the intersection of any two parallel paths in $\mathcal{P}$ consists of exactly one subpath. Since $P_i$ and $Q$ are parallel, this implies that $Q$ contains the subpath $P_i[y,z]$. In particular, $Q$ contains each of $P_i[x_{p+1}, x_{p+2}], \dots, P_i[x_{q-1},x_q]$, so $Q$ does not overlap any of the bigons $B_{p+1}, \dots, B_{q-1}$. It follows that $Q$ overlaps at most two bigons. 
        
        For the second case, a symmetric argument shows that if $Q$ is parallel to $P_j$ then $Q$ overlaps at most two bigons. 
        For the third case, if $Q$ is opposite to both $P_i$ and $P_j$, then by Lemma~\ref{L:opposite}, $Q$ is edge-disjoint from both $P_i$ and $P_j$, and so does not overlap any bigons.
    \end{proof} 

    \begin{proof}[Proof of Lemma~\ref{L:fact-two}]
        Suppose for the sake of argument that $B_p$ is an odd bigon that does not overlap any path. The bigon $B_p$ is below $P_i$ and above $P_j$. Specifically, it is bounded by $P_i[x_p,x_{p+1}] \cup P_j[x_{p+1}, x_p]$. 
        To lighten notation, let $A = P_i[x_p,x_{p+1}]$ and let $B = P_j[x_{p+1}, x_p]$. 
        Our goal is to reduce the number of crossings in $\mathcal{P}$ via an exchange procedure similar to those used in previous lemmas. Roughly speaking, we will do this by reversing the orientations of $A$ and $B$ and by modifying the paths that enter $B_p$ so that they no longer do so.
    
        First we describe how to reverse the orientations of $A$ and $B$.
        By assumption, all paths in $\mathcal{P}$ that use edges in $A$ must contain $A$, and all paths in $\mathcal{P}$ that use edges in $B$ must contain $B$. Let $\mathcal{Q}_L$ be the set of paths that contain $A$ and let $\mathcal{Q}_R$ be the set of paths that contain $L$. By Lemma~\ref{L:opposite}, all paths in $\mathcal{Q}_L$ are parallel to $P_i$ and all paths in $\mathcal{Q}_R$ are parallel to $P_j$. Now we simply let 
        \[
            P_l' = P_l[s_l, x_p] \circ rev(B) \circ P_l[x_{p+1}, t_l]
        \]
        for any path $P_l \in \mathcal{Q}_l$, and let 
        \[
            P_r' = P_r[s_r, x_{p+1}] \circ rev(A) \circ P_r[x_p, t_r]
        \]
        for any path $P_r \in \mathcal{Q}_R$. Let $\mathcal{Q}_L' = \{P_l' | P_l \in \mathcal{Q}_L\}$ and $\mathcal{Q}_R' = \{P_r' | P_r \in \mathcal{Q}_R\}$. Note that $P_i \in \mathcal{Q}_L$ and $P_j \in \mathcal{Q}_R$, so we have described how to modify $P_i$ and $P_j$.
        
        Now we describe how to modify the paths that enter $B_p$. 
        This is necessary so that the paths do not cross with the paths $P_l'$ and $P_r'$ described in the previous paragraph.
        By Lemma~\ref{L:opposite}, $A$ only crosses paths parallel to $P_j$, and $B$ only crosses paths parallel to $P_i$. Let $\mathcal{Q}_A$ be the set of paths that cross $A$, and let $\mathcal{Q}_B$ be the set of paths that cross $B$. For each path $P_a \in \mathcal{Q}_A$, let $u_a$ be the first vertex (of $P_a$) at which $P_a$ touches $A$ and let $v_a$ be the last. We define 
        \[
            P_a' = P_a[s_a, u_a] \circ A[u_a,v_a] \circ P_a[v_a, t_a].
        \]
        Note that $P_a'$ conflicts with $A$ but does not conflict with $rev(A)$. Furthermore, $P_a'$ no longer crosses $A$. Similarly, for each path $P_b \in \mathcal{Q}_B$, let $u_b$ be the first vertex (of $P_b$) at which $P_b$ crosses $B$, let $v_b$ be the last vertex at which $P_b$ crosses $B$, and let 
        \[
            P_b' = P_b[s_b, u_b] \circ B[u_b,v_b] \circ P_b[v_b, t_b].
        \]
        Let $\mathcal{Q}_A' = \{P_a' | P_a \in \mathcal{Q}_A\}$ and $\mathcal{Q}_B' = \{P_b' | P_b \in \mathcal{Q}_B\}$. This finishes the description of how to modify $\mathcal{P}$ to reduce the number of crossings. That is, let
        \[
            \mathcal{P}' = \mathcal{P} \setminus (\mathcal{Q}_L \cup \mathcal{Q}_R \cup \mathcal{Q}_A \cup \mathcal{Q}_B) \cup (\mathcal{Q}_L' \cup \mathcal{Q}_R' \cup \mathcal{Q}_A' \cup \mathcal{Q}_B').
        \]
        We need to show that $\mathcal{P}'$ is a solution with fewer crossings than $\mathcal{P}$.
        Paths $A$ and $B$ are shortest paths, so all subpaths of $A$ and $B$ are shortest paths and all paths in $\mathcal{P}'$ are shortest paths.
        All paths in $\mathcal{P}'$ use the edges of $A$ in the reverse direction (i.e., from $x_{p+1}$ to $x_p$), if at all. Similarly, all paths in $\mathcal{P}'$ use the edges of $B$ in the reverse direction (i.e., from $x_p$ to $x_{p+1}$), if at all. All arcs used by $\mathcal{P'}$ that are not used by $\mathcal{P}$ are in $A$ or $B$, so this implies that the paths in $\mathcal{P}'$ are non-conflicting. During the exchange procedure, we replace subpaths in or on $B_p$ with subpaths of the boundary of $B_p$, such that no paths in $\mathcal{P}'$ enter $B_p$. Tedious casework implies that no crossings are added when we go from $\mathcal{P}$ to $\mathcal{P}'$; for details, see the appendix. Without increasing the number of crossings in $\mathcal{P}'$, we can also use the procedure in the proof of Lemma~\ref{L:connected-intersection} to modify the paths in $\mathcal{P}'$ so that the intersection of any pair of parallel paths consists of a single subpath. 
        On the other hand, given any pair of paths $P_l \in \mathcal{Q}_L$ and $P_r \in \mathcal{Q}_R$, $P_l'$ and $P_r'$ have strictly fewer crossings than $P_l$ and $P_r$; for details, see the appendix again. This contradicts the fact that $\mathcal{P}$ has the fewest crossings out of all solutions that satisfy Lemma~\ref{L:connected-intersection}. We have thus proved the lemma for odd bigons.
        A symmetric argument proves the lemma for even bigons.

    \end{proof}
    
    \begin{figure}
    \centering
    \begin{tabular}{cr@{\qquad}cr@{\qquad}cr}
     	\includegraphics[scale = 0.45]{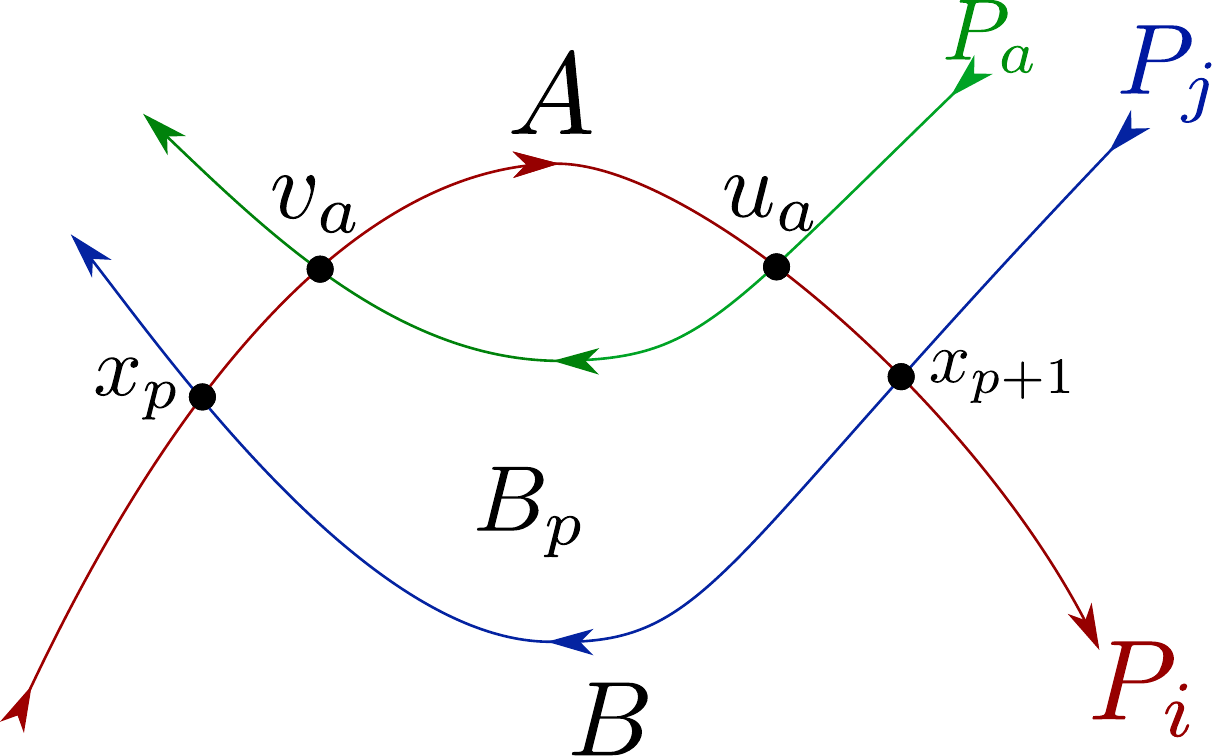} & \hspace{-0.25in} 
    \end{tabular}
    	
    \caption{Paths before the exchange procedure, in the proof of Lemma~\ref{L:fact-two}. $P_i$ is red, $P_j$ is blue, and $P_a$ is green. $B_p$ is bounded by $A$ and $B$}
    \label{F:few-crossing-parallel}
    \end{figure}

\section{NP-hardness of Ideal Orientation}
In this section we show that the orientation problem is NP-hard in unweighted planar graphs when $k$ is part of the input. The reduction is from planar 3-SAT and is similar to reductions by Middendorf and Pfeiffer~\cite{MP} and by Eilam-Tzoreff~\cite{E-T98}. Planar 3-SAT is the special case of 3-SAT where a certain bipartite graph $G(y)$ is planar, defined as follows. Given an instance $y$ of 3-SAT, each variable of $y$ is a vertex, and each clause of $y$ is also a vertex. For every variable $x_i$ and every clause $c_j$, we add an edge between $x_i$ and $c_j$ if either $x_i$ or $\overline{x_i}$ appears in $c_j$. The resulting graph $G(y)$ is bipartite; if it is planar, then $y$ is an instance of planar 3-SAT. Lichtenstein showed that planar 3-SAT is still NP-hard~\cite{L82}.

Suppose we are given an instance $y$ of planar 3-SAT. As noted by Middendorf and Pfeiffer~\cite{MP}, we may assume that each variable appears in three clauses. To see this, fix a planar embedding of $G(y)$, and let $vC_1, \dots, vC_k$ be the edges incident to a variable $v$ in clockwise order. Introduce new variables $v_1, \dots, v_k$ and clauses $v_k \vee \neg v_1$ and $v_i \vee \neg v_{i+1}$ for all $i \in \{1, \dots, k-1\}$. In addition, replace the occurrence of $v$ in $C_i$ with $v_i$. If we do this for all variables $v$, we get an instance $y'$ of planar 3-SAT that is satisfiable if and only if $y$ is satisfiable, and every variable in $y'$ appears in exactly three clauses.

We use $y$ to construct an instance of the ideal orientation problem. We will construct a clause gadget for each clauses and a variable gadget for each variable.  
The clause gadget for a clause $C$ is shown in Figure~\ref{F:clause}. There are three terminals pairs $(s_C, t_C), (s_C', t_C'),$ and $(s_C'', t_C'')$.  Let us note some key properties of $G_C$. We have $d(s_C, t_C) = d(s_C'', t_C'') = 3$ and $d(s_C', t_C') = 4$.
There are two shortest paths from $s_C$ to $t_C$, three shortest paths from $s_C'$ to $t_C'$, and two shortest paths from $s_C''$ to $t_C''$. There exist pairwise non-conflicting shortest paths connecting $(s_C, t_C), (s_C', t_C'),$ and $(s_C'', t_C'')$ in $G_C$. These paths must use at least one of the edges $ab = e_{vC}$, $cd = e_{wC}$, and $ef = e_{xC}$. Furthermore, three such non-conflicting shortest paths exist even when two of the three edges are not to be used. 

\begin{figure}
\centering
	\includegraphics[scale=0.3]{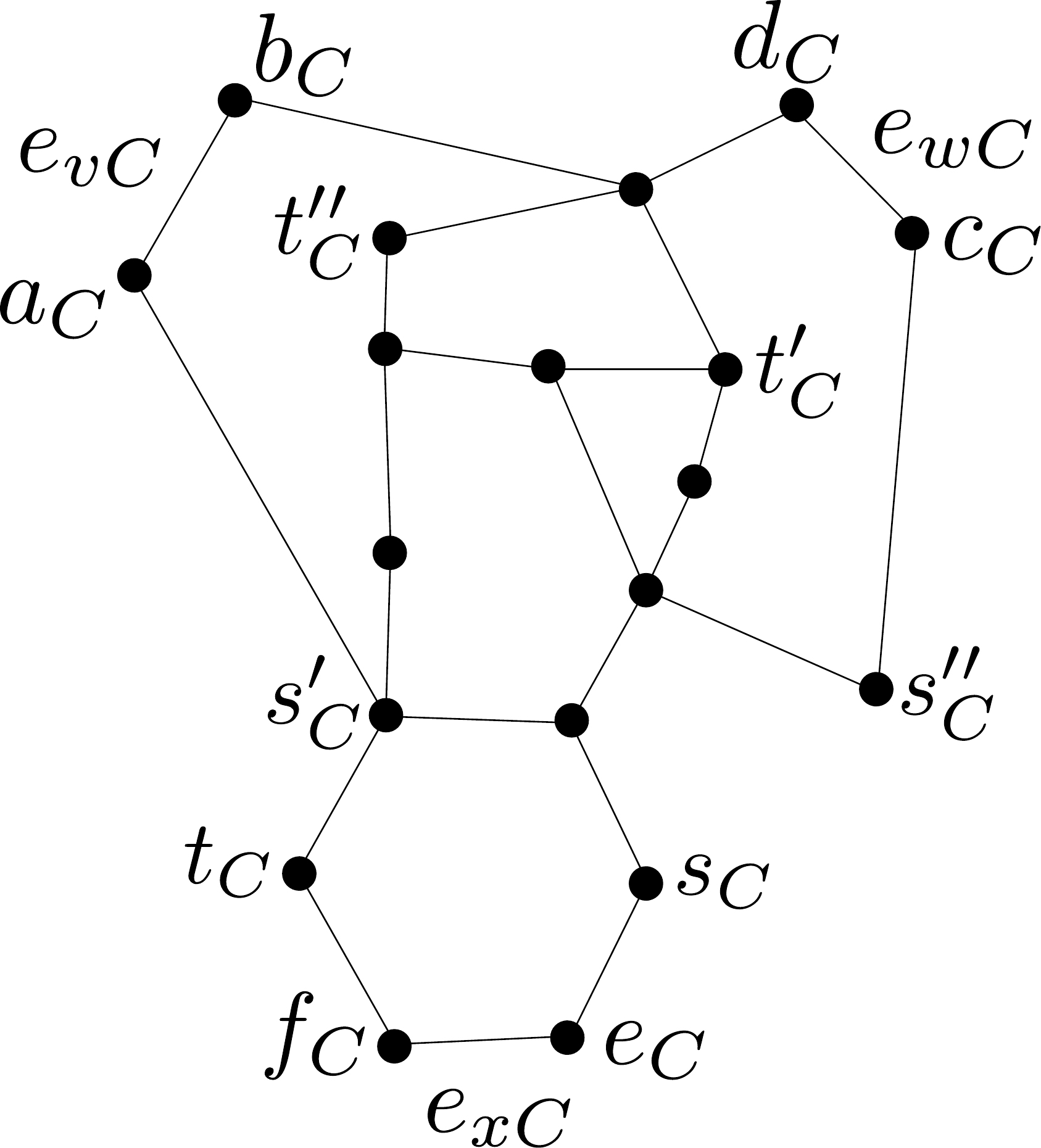}
\caption{Clause gadget $G_C$ for a clause $C$ containing variables $v$, $w$, $x$. All edges are unweighted.}
\label{F:clause}
\end{figure}

The edges $e_{vC}$, $e_{wC}$, and $e{xC}$ are part of the clause gadget associated with $C$ and will also each be in variable gadgets associated with $v$, $w$, and $x$, respectively. Before defining the variable gadgets, we need to fix some terminology regarding the orientations of the three edges.
Each of the three edges can be oriented {\em forward} or {\em backward} as follows. The forward orientation of $e_{vC}$ is from $a$ to $b$, the forward orientation of $e_{wC}$ is from $c$ to $d$, and the forward orientation of $e_{xC}$ is from $e$ to $f$. The backward orientation of an edge is simply the reverse of the forward orientation.  Intuitively, an edge must be oriented forward in order to be used in some shortest path connecting a pair of terminals; furthermore, orienting an edge $e_{vC}$ forward means that the literal $v$ or $\neg v$ (whichever one appears in $C$) is set to True.

We also give each of the three edges a {\em true} and a {\em false} orientation depending on whether the literals in $C$ are positive or negative. If $v$ is a literal in the clause $C$, then the true orientation of $e_{vC}$ is the forward orientation of $e_{vC}$ and the false orientation of $e_{vC}$ is the backward orientation. If $\neg v$ is a literal in $C$, then the true orientation of $e_{vC}$ is the backward orientation and the false orientation is the forward orientation. True and false orientations for $e_{wC}$ and $e_{xC}$ are defined analogously.
Intuitively, the true orientation of an edge $e_{vC}$ is the direction that it would be oriented in if the variable $v$ were assigned to true.

Finally, each of the three edges has a {\em clockwise} orientation and a {\em counterclockwise} orientation. The clockwise orientation of $e_{vC}$ is its forward orientation, the clockwise orientation of $e_{xC}$ is its forward orientation, and the clockwise orientation of $e_{wC}$ is its {\em backward} orientation.
The counterclockwise orientation of an edge is the reverse of its clockwise orientation. Intuitively, an edge oriented clockwise goes clockwise around its clause gadget, and an edge oriented counterclockwise goes counterclockwise around its clause gadget. However, somewhat confusingly, we will construct the variable gadgets such that a clockwise-oriented edge goes {\em counterclockwise} around its {\em variable} gadget and a counterclockwise-oriented edge goes {\em clockwise} around its {\em variable} gadget.

For each variable $v$, we construct a variable gadget $G_v$ as follows. Suppose $v$ appears in clauses $C, D,$ and $E$; suppose further that $vC, vD,$ and $vE$ are the edges incident to $v$ in clockwise order in $G(y)$. For each of the three edges $e_{vC}$, $e_{vD}$, and $e_{vE}$ (in the clause gadgets), we check whether or not the true orientation of the edge is the counterclockwise orientation of that edge. There are four cases:
\begin{itemize}
    \item If for each of the three edges the true orientation is the counterclockwise orientation, then we construct the variable gadget in Figure~\ref{F:variable}a. 
    \item If for exactly two of the three edges (without loss of generality, $(v,C)$ and $(v,D)$) the true orientation is the counterclockwise orientation, then we construct the variable gadget in Figure~\ref{F:variable}b. 
    \item If for exactly one of the three edges (without loss of generality, $(v,E)$) the true orientation is the counterclockwise orientation, then again we still construct the variable gadget in Figure~\ref{F:variable}b. 
    \item If for each of the three edges the true orientation is the clockwise orientation, then we still construct the variable gadget in Figure~\ref{F:variable}a. 
\end{itemize}
Finally, for every variable $v$ and every clause $C$ we identify the edge $e_{vC}$ in both $G_v$ and $G_C$. The resulting graph is still planar and is $G_1(y)$.

\begin{figure}
\centering
\begin{tabular}{cr@{\qquad}cr@{\qquad}cr}
	\includegraphics[scale = 0.25]{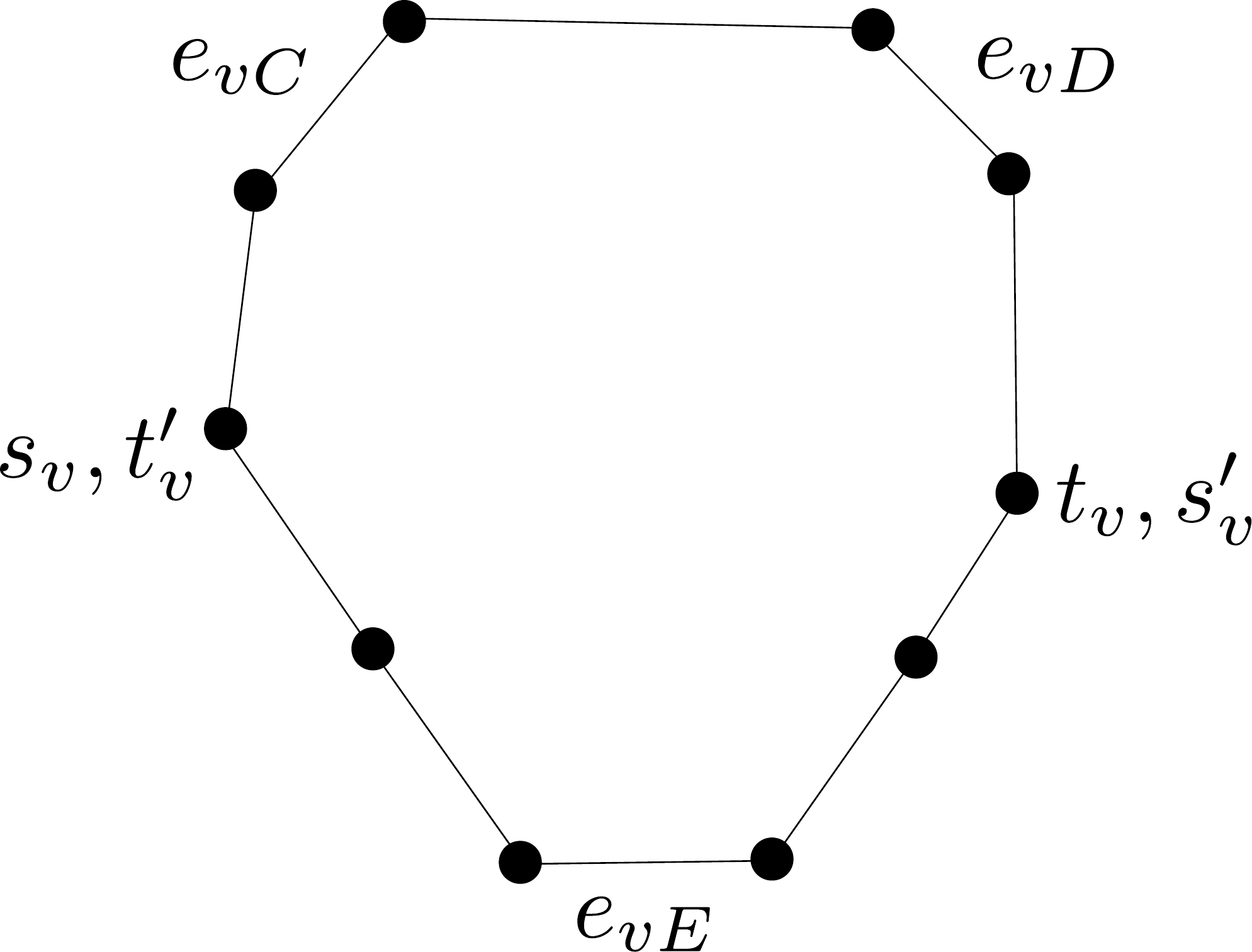} & \hspace{-0.25in} (a)
	&
	\includegraphics[scale = 0.25]{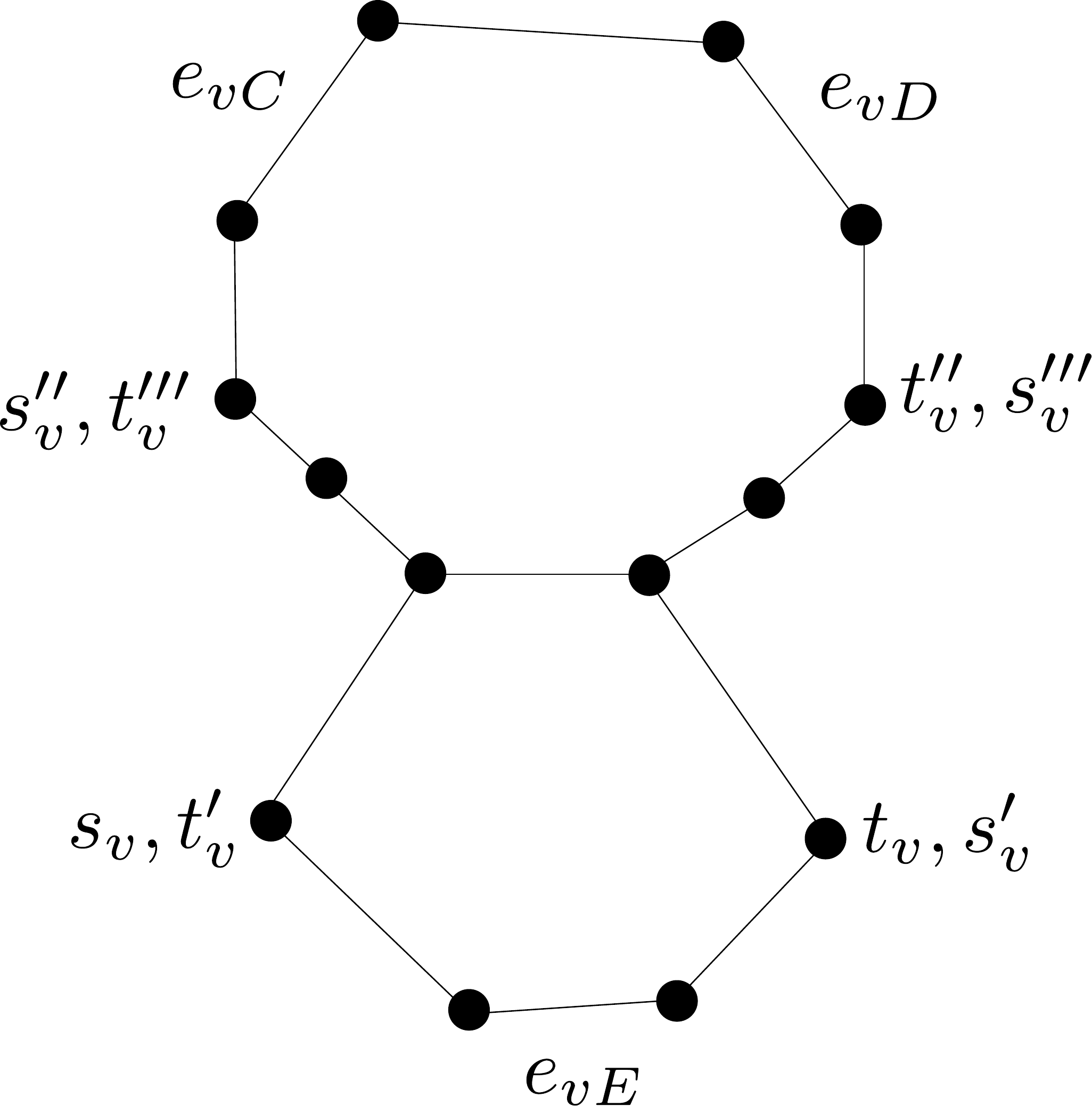} & \hspace{-0.25in} (b)
\end{tabular}
\caption{possible variable gadgets $G_v$ for a variable $v$ appearing in three clauses $C, D,$ and $E$}
\label{F:variable}
\end{figure}

$G_v$ is constructed so that there are only two ways to orient the edges. In one orientation, all edges $vC, vD,$ and $vE$ are oriented in the true direction, and in the other orientation, the three edges are oriented in the false direction. Orienting the three edges in the true direction corresponds to setting the variable to True, and orienting them in the false direction corresponds to setting them to False. The reduction clearly takes polynomial time, and the following lemma implies its correctness.

\begin{lemma}
    A planar 3-SAT formula $y$ is satisfiable if and only if there exists an ideal orientation in $G_1(y)$.
\end{lemma}
\begin{proof}
    \underline{$\Rightarrow: $} Suppose $y$ is satisfiable, and fix a satisfying assignment. For each clause $C$, we orient the edges in $G_C$ as follows. For each of the three literals, we do the following. Let $v$ or $\neg v$ be some literal in $C$. Orient the edge $e_{vC}$ forwards if $v$ or $\neg v$ is in $C$ and set to True; otherwise, the edge is oriented backwards. We know that exactly one of the three edges $e_{vC}, e_{wC}, e_{xC}$ is oriented forwards. It is possible to orient the rest of the edges in $G_C$ such that distances between the terminal pairs $(s_C, t_C), (s_C', t_C')$, and $(s_C'', t_C'')$ are preserved. 
    
    In each variable gadgets $G_v$, we orient the edges as follows. If $v$ is set to True, then each of $e_{vC}, e_{vD}, e_{vE}$ are oriented in the true direction; otherwise, the three edges are oriented in the false direction. It is possible to orient the rest of the edges in $G_v$ such that the distances between the terminal pairs $(s_v, t_v), (s_v', t_v'$, $(s_v'', t_v'')$, and $(s_v''', t_v''')$ (if they exist) are preserved. To show that orientations are consistent, recall that in the clause gadgets, there are four cases:
    \begin{itemize}
        \item $v$ appears in $C$ and is set to True. Then $e_{vC}$ is oriented forward and so is oriented in the true direction.
        \item $\neg v$ appears in $C$ and $v$ is set to False. Then $e_{vC}$ is oriented forward and in the false direction.
        \item $v$ appears in $C$ and is set to False. Then $e_{vC}$ is oriented backward and so is- oriented in the false direction.
        \item $\neg v$ appears in $C$ and $v$ is set to True. Then $e_{vC}$ is oriented backward and so is oriented in the true direction.
    \end{itemize}
    In all cases we see that the orientation in the clause gadget is consistent with the orientation in the variable gadget.
    
    \underline{$\Leftarrow$:} Suppose an ideal orientation exists. If $e_{vC}, e_{vD},$ and $e_{vE}$ are all oriented in the true direction, then set $v$ to True; otherwise they are all oriented in the false direction and we set $v$ to False. We need to show that this is a satisfying assignment. Consider a clause $C$. Since an ideal orientation exists, at least one of the edges $e_{vC}$, $e_{wC}$, and $e_{xC}$  must be oriented forward. Say $e_{vC}$ is oriented forward. This means that either $e_{vC}$ is oriented in the true direction with $v$ appearing positively, or $e_{vC}$ is oriented in the false direction with $v$ appearing negatively. In the first case, $v$ is set to True, so $C$ is satisfied. In the second case, $v$ is set to False, so $C$ is satisfied.
\end{proof}

\section{Serial Case for $k$-min-sum orientations}
In this section, we describe an algorithm to solve serial instances of the $k$-min-sum orientation problem.
Recall that every terminal pair $(s_i, t_i)$ is either clockwise (i.e., a clockwise traversal of the outer face will visit $s_i$ and then immediately visit $t_i$) or counterclockwise. 
Given a set $\Pi$ of arbitrary directed paths $\pi_1, \dots, \pi_k$ such that $\pi_i$ connects $s_i$ to $t_i$, we define ``lower envelopes'' $L(1, \Pi), \dots, L(k, \Pi)$ in the same way as in section~\ref{S:serial-ideal}.
To simplify our presentation, we assume that our given instance has a unique solution, if it exists.  If necessary, this uniqueness assumptions can be enforced with high probability using the isolation lemma of Mulmuley, Vazirani, and Vazirani~\cite{MVV}. 

Before we describe the algorithm, we prove an analog of Lemma~\ref{L:no-crossings}.

\begin{lemma}\label{L:no-crossing-sum}
    Let $G$ be a serial instance of the $k$-min-sum orientation problem with terminal pairs $(s_1, t_1)$, $\dots$, $(s_k, t_k)$. If a solution exists, then the paths in the solution are pairwise non-crossing.
\end{lemma}
\begin{proof}
    The proof is similar to that of Lemma~\ref{L:no-crossings}. Let $\mathcal{P} = \{P_1, \dots, P_k\}$ be the unique solution to the serial instance $G$ of the $k$-min-sum orientation problem, where $P_i$ connects $s_i$ to $t_i$. The walks $L(1, \mathcal{P}), \dots, L(k, \mathcal{P})$ are pairwise non-crossing. By Lemma~\ref{L:envelope-nonconflicting} they are pairwise non-conflicting. By Lemma~\ref{L:envelope-length} and our uniqueness assumption, their total length is strictly less than that of the paths in $\mathcal{P}$. This contradicts the fact that $\mathcal{P}$ was the optimal solution to $G$.
\end{proof}

    Let $\mathcal{P} = \{P_1, \dots, P_k\}$ be the unique solution to the instance $G$ of the $k$-min-sum solution.
    By the Jordan Curve Theorem, non-crossing opposite paths must be edge-disjoint.
    This suggests the following algorithm, which occurs in two phases.
    
    \begin{enumerate}
        \item In the first phase, we re-index the terminals so that $(s_1,t_1), \dots, (s_\alpha, t_\alpha)$ are clockwise and $(s_{\alpha+1}, t_{\alpha+1})$ are counterclockwise. We split the instance of the $k$-min-sum problem into two sub-instances. One of the sub-instances consists of the original graph $G$ with the clockwise terminal pairs, while the other sub-instance consists of $G$ with the counterclockwise terminal pairs. We solve each sub-instance separately. In the clockwise sub-instance, we are finding $\alpha$ non-crossing edge-disjoint directed paths of minimum total length such that the $i$-th path connects $s_i$ to $t_i$ for $i \in \{1, \dots, \alpha\}$ (We will describe later how to find such paths). Likewise, in the counterclockwise sub-instance we are finding $k - \alpha$ non-crossing edge-disjoint directed paths of minimum total length such that the $(i-\alpha)$-th path connects $s_i$ to $t_i$ for $i \in \{\alpha + 1, \dots, k\}$. We then let $\Pi = \{\pi_1, \dots, \pi_k\}$ be the set of all $k$ paths, where $\pi_i$ connects $s_i$ to $t_i$. By Lemma~\ref{L:no-crossing-sum}, the sum of the lengths of $\pi_1, \dots, \pi_k$ is at most the sum of the lengths of the paths in $\mathcal{P}$. 
        \item Any two opposite paths in $\Pi$ are edge-disjoint and so are non-conflicting. However, parallel paths (i.e., a clockwise path and a counterclockwise path) found by the first phase may conflict with each other; the purpose of phase 2 is to remove these conflicts. In phase 2, we simply output $L(1, \Pi), \dots, L(k, \Pi)$. By Lemma~\ref{L:no-crossing-sum}, the sum of the lengths of the output paths is no greater than the sum of the lengths of the paths in $\mathcal{P}$. Since the output paths are envelopes, they non-crossing; the Jordan Curve Theorem then implies that two output paths can conflict only if they are opposite. On the other hand, Lemma~\ref{L:envelope-nonconflicting} implies that opposite paths are non-conflicting.
        Thus the output paths are indeed non-conflicting paths of minimum total length that connect the terminals.
    \end{enumerate}
    
    To finish the description of the algorithm we just need to show how to find the non-crossing edge-disjoint directed paths in Phase 1. Before doing this, we define the {\em $k$-min-sum non-crossing edge-disjoint paths problem ($k$-NEPP)} and the {\em $k$-min-sum vertex-disjoint paths problem ($k$-VPP)} as follows. In $k$-NEPP we are given a plane graph $G$ with $k$ pairs of terminals $(s_1, t_1), \dots, (s_k, t_k)$, and we wish to find $k$ paths $P_1, \dots, P_k$ such that $P_i$ connects $s_i$ to $t_i$ and the $k$ paths are pairwise non-crossing and edge-disjoint. (Note that under our definition of ``edge-disjoint,'' finding edge-disjoint directed paths in undirected graphs is the same as finding edge-disjoint undirected paths in undirected graphs. Thus for the rest of this section all paths will be undirected.) $k$-VPP is similar except that the paths $P_1, \dots, P_k$ are to be vertex-disjoint instead of non-crossing edge-disjoint.
    It is known that $k$-VPP can be solved in serial instances in $O(kn^5)$ time when edge lengths are non-negative~\cite{BNZ15}.
    
    In order to find the paths in Phase 1 we need to solve serial instances of $k$-NEPP. We will solve such instances by reducing to serial instances of $k$-VPP; this will finish the description of the algorithm for serial instances of the $k$-min-sum orientation problem.
    
    The reduction is as follows. Starting with $G$, we replace each vertex $v$ in $G$ with an undirected cycle $C_v$ of $\deg(v)$ vertices $v_1, \dots, v_{\deg(v)}$. Each edge in the cycle has length zero. 
    We make every edge that was incident to $v$ incident to some vertex $v_i$ instead, such that each edge is connected to a different vertex $v_i$, the clockwise order of the edges is preserved, and the graph remains planar. 
    The resulting graph $G^\circ$ has $O(n)$ vertices and arcs.
See Figure~\ref{F:cycle-reduction}.
    Furthermore, if $G$ has all terminals on the outer face, then so does $G^\circ$.    
    
    \begin{figure}
\centering
\begin{tabular}{crcr}
	\includegraphics[scale=0.38]{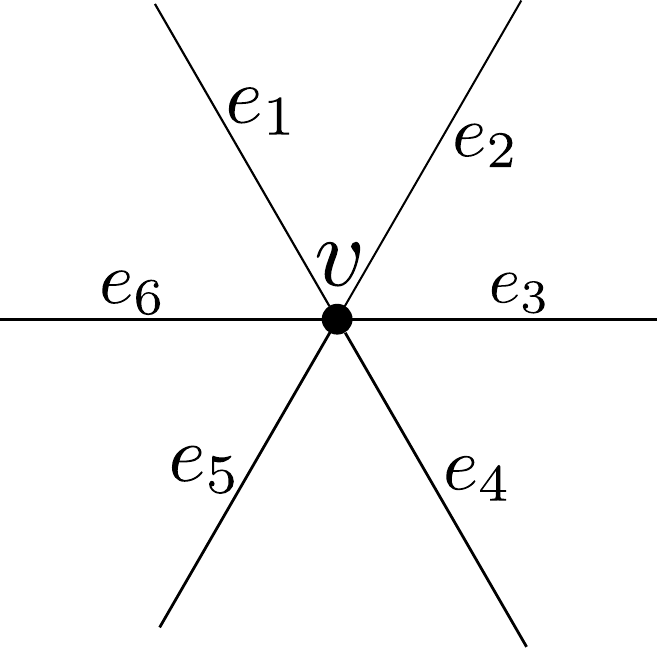} & \hspace{-0.25in}(a)
	&
	\includegraphics[scale=0.38]{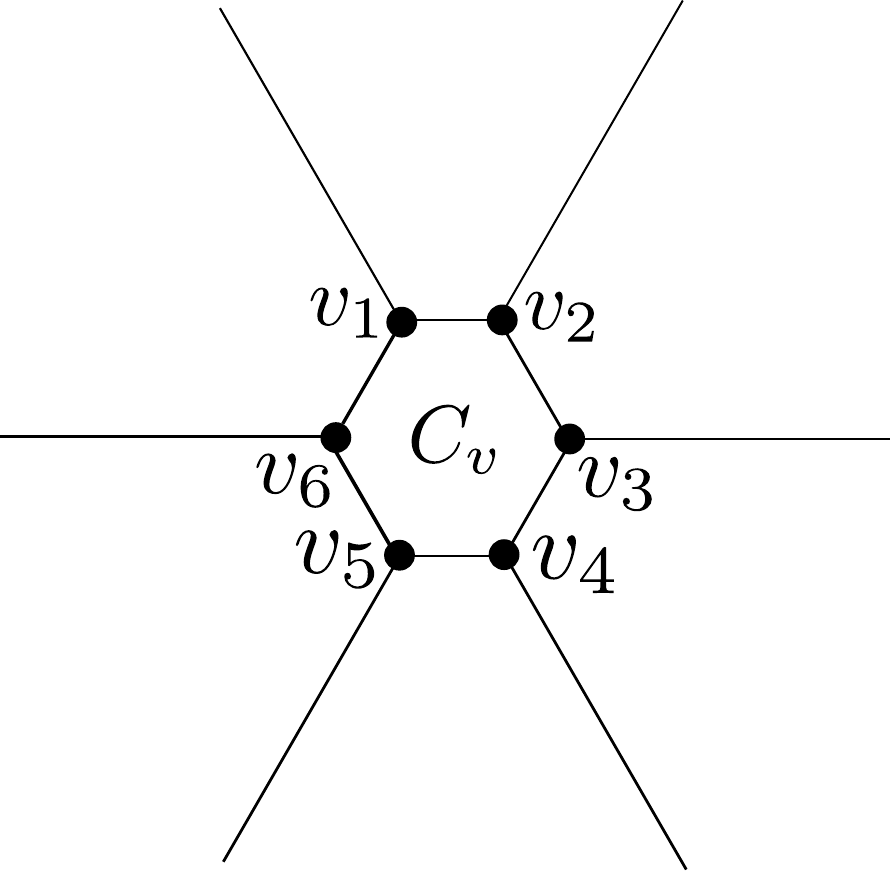} & \hspace{-0.25in}(b)
\end{tabular}
\caption{(a) vertex $v \in V(G)$ with incident edges $e_1, \dots, e_6$ (b) corresponding cycle $C_v$ in $G^\circ$; each edge in $C_v$ has zero length}
\label{F:cycle-reduction}
\end{figure}
    
    \begin{lemma}
        Suppose $G$ is serial instance with terminal pairs $(s_1, t_1), \dots, (s_k, t_k)$. The following statements are equivalent:
        \begin{enumerate}
            \item There exist pairwise non-crossing edge-disjoint paths $P_1, \dots, P_k$ of total length $L$ in $G$ such that $P_i$ connects $s_i$ and $t_i$ for all $i$.
            \item There exist pairwise vertex-disjoint paths $Q_1, \dots, Q_k$ of total length $L$ in $G^\circ$ such that $Q_i$ connects $s_i$ and $t_i$ for all $i$.
        \end{enumerate}
    \end{lemma}
    \begin{proof}
        \underline{$\Rightarrow:$} Suppose there exist pairwise non-crossing edge-disjoint paths $P_1, \dots, P_k$ of total length $L$ in $G$ such that $P_i$ connects $s_i$ and $t_i$. We construct the paths $Q_1, \dots, Q_k$ as follows. For any edge $e$ in $P_i$, we add $e$ to $Q_i$. This defines the portions of the paths $Q_1, \dots, Q_k$ outside the cycles $C_v$; these portions are vertex-disjoint because by construction the endpoints of edges of $G$ are all distinct in $G^\circ$. 
        
        We route the portions of $Q_1, \dots, Q_k$ inside the cycles $C_v$ in $G^\circ$ as follows. Let $v$ be a vertex of $G$, and suppose $P_i$ go through $v$. Suppose the cyclic order of the edges around $v$ is $e_1, \dots, e_d$, where $d = \deg(v)$. Say $P_i$ goes into $v$ through $e_x$ and leaves through $e_y$, where $x<y$. By the Jordan Curve Theorem, either no other path uses $e_{x+1}, \dots, e_{y-1}$ or no other path uses $e_{y+1}, \dots, e_d, e_1, \dots, e_{x-1}$. Suppose the first case holds (the second case is symmetric). Route the path $Q_i$ through vertices $v_x, \dots, v_y$. The resulting paths $Q_1, \dots, Q_k$ are vertex-disjoint from because none of the paths $P_1, \dots, P_k$ use $e_{x+1}, \dots, e_{y-1}$ (except possibly $P_i$). Clearly $Q_1, \dots, Q_k$ have the same length as $P_1, \dots, P_k$.
        
        \underline{$\Leftarrow:$} Suppose there exist pairwise vertex-disjoint paths $Q_1, \dots, Q_k$ of total length $L$ in $G^\circ$. Trivially, the paths $Q_1, \dots, Q_k$ are pairwise non-crossing edge-disjoint too. Each path $P_i$ can be defined by ``projecting'' $Q_i$ into $G$ in the obvious way: an edge of $G$ is in $P_i$ if and only if $e$ was in the original path $Q_i$. The resulting paths $P_1, \dots, P_k$ are pairwise non-crossing because the original paths $Q_1, \dots, Q_k$ were pairwise non-crossing. 
        By similar reasoning as in the second half of the proof of Lemma~\ref{L:grid-correctness}, we can see that the paths $P_1, \dots, P_k$ are pairwise non-crossing and edge-disjoint. Clearly $P_1, \dots, P_k$ are the same length as $Q_1, \dots, Q_k$.
    \end{proof}
    We can use the algorithm of Borradaile, Nayyeri, and Zafarani~\cite{BNZ15} to solve serial instances of $k$-min-sum vertex-disjoint paths. Since $G^\circ$ has $k$ pairs of terminals and $O(n)$ vertices and edges, the algorithm of Borradaile, Nayyeri, and Zafarani still takes $O(kn^5)$ time to compute $\Pi$. Given $\Pi$, computing the envelopes $L(1, \Pi), \dots, L(k, \Pi)$ takes $O(n)$ time, so our entire algorithm still takes $O(kn^5)$ time.

\bibliographystyle{plainurl}
\bibliography{references}{}

\appendix
\section{Omitted proofs}\label{A:LL17}
\begin{proof}[Proof of Lemma~\ref{L:uncross-codirectional}]  
    Let $\mathcal{P} = \{P_1, \dots, P_k\}$ be a solution to the 
    Re-index the terminal pairs such that $\{P_1, \dots, P_h\}$ are pairwise parallel paths and in fact form an equivalence class.
    It is straightforward to verify that the clockwise order of the terminals is then $s_1, \dots, s_h, t_h, \dots t_1$.
    
    We uncross the paths inductively. Suppose paths $P_1, \dots, P_{i-1}$ are uncrossed but $P_i$ crosses $P_{i-1}$.
    Let $x$ be the first vertex of $P_i$ on $P_{i-1}$ and let $y$ be the last. see Figure~\ref{F:uncross-serial}a. 
	Now we exchange $P_i[x,y]$ for $P_{i-1}[x,y]$. 
	In other words, let $P_i' = P_i[s_i, x] \circ P_{i-1}[x,y] \circ P_i[y, t_i]$, and let $\mathcal{P}' = \mathcal{P} \setminus \{P_i\} \cup \{P_i'\}$. 
	Since $P_i[x,y]$ and $P_{i-1}[x,y]$ are shortest paths, $P_i'$ is still a shortest path connecting $s_i$ to $t_i$. Since $P_i'$ only uses arcs in $P_i$ and $P_{i-1}$, $P_i'$ does not conflict with any other path in $\mathcal{P}'$.
	Thus $\mathcal{P}'$ is another set of $k$ non-conflicting shortest paths. Furthermore, paths $\{P_1, \dots, P_{i-1}, P_i'\}$ are pairwise non-crossing.
	
	    \begin{figure}
    \centering
    \begin{tabular}{cr@{\qquad}cr@{\qquad}cr}
    	\includegraphics[scale = 0.35]{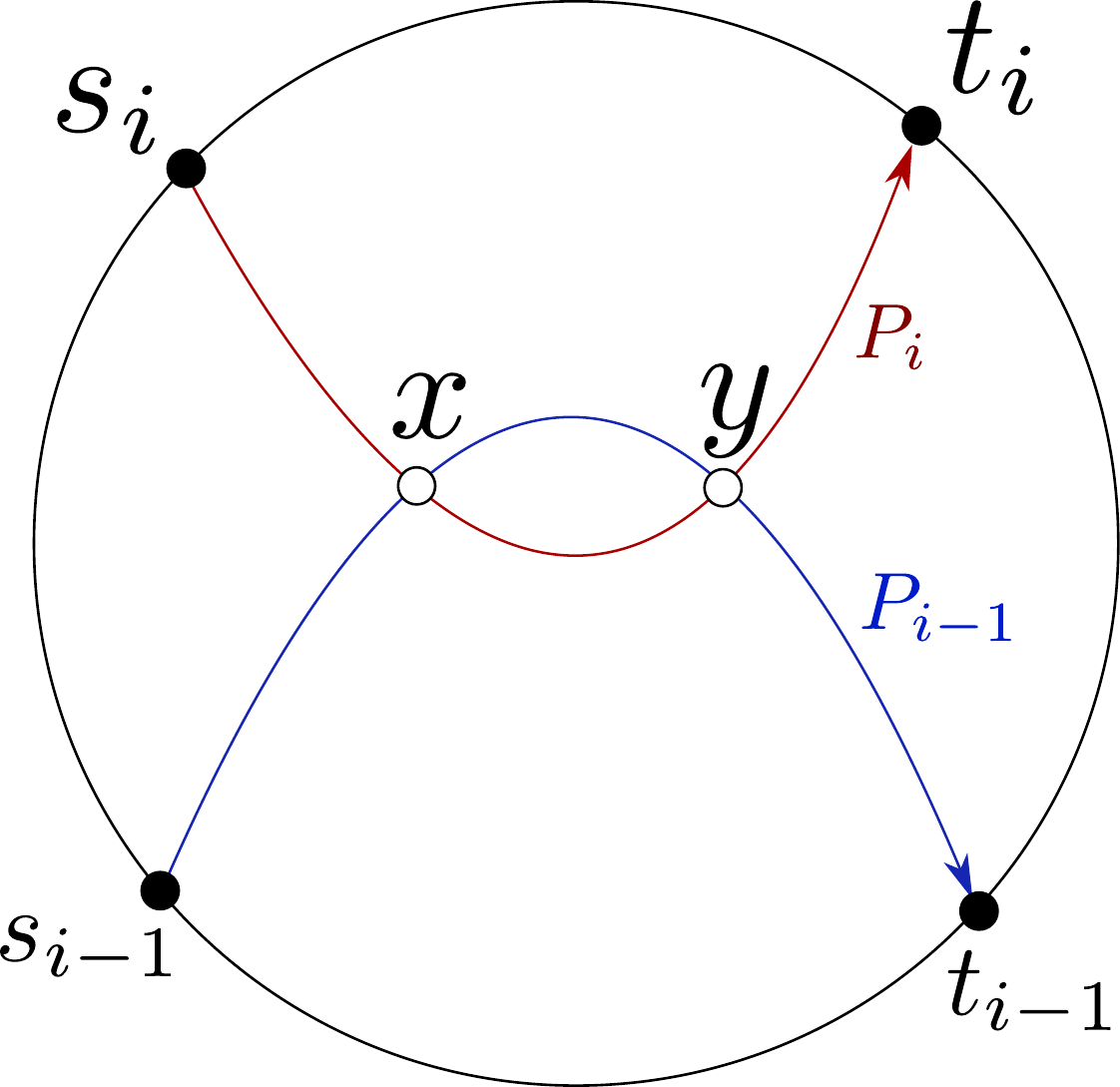} & \hspace{-0.25in} 
    \end{tabular}
    \caption{Uncrossing $P_i$ and $P_{i-1}$: replace $P_i$ with $P_i[s_i, x] \circ P_{i-1}[x, y] \circ P_i[y,t_i]$}
    \label{F:uncross-serial}
    \end{figure}
	
	This shows how to uncross paths of one equivalence class. Simply repeat for all other equivalence classes.
\end{proof}
\begin{proof}[Rest of the proof of Lemma~\ref{L:fact-two}]
        We need to extend the definitions of ``below'' and ``above'' introduced in subsection~\ref{SS:crossing-bound}.
        Suppose $P$ and $Q$ are paths in $G$ whose endpoints are on $\partial G$. Suppose further that the endpoints of any two of $P, Q,$ and $P_i$ do not cross. Let $C_i$ be the portion of $\partial G$ from $s_i$ to $t_i$ that does not contain $s_j$ or $t_j$. There are two cases.
        \begin{enumerate}
            \item Suppose the endpoints of $Q$ are not in $C_i$. The path $Q$ divides the interior of $G$ into two regions. If $P$ lies entirely in the region whose closure contains $s_i$ and $t_i$, then $P$ is below $Q$.
            \item Suppose the endpoints of $Q$ are in $C_i$. The path $Q$ divides the interior of $G$ into two regions. If $P$ lies entirely in the region whose closure does not contain $s_j$ and $t_j$, then $P$ is below $Q$. 
        \end{enumerate}

        Now let $P$ and $Q$ be paths in $\mathcal{P}$. 
        To simplify notation, let $P' = P$ if $P \notin \mathcal{Q}_L \cup \mathcal{Q}_R \cup \mathcal{Q}_A \cup \mathcal{Q}_B$, so that $\mathcal{P}' = \{p' | p \in \mathcal{P}\}$. 
        
        First we will show that $P'$ and $Q'$ do not cross more times than $P$ and $Q$ cross. There are eight different cases (not counting symmetric cases). For the first four cases, suppose $P$ and $Q$ are both parallel to $P_j$, so that $P$ and $Q$ do not cross by Lemma~\ref{L:uncross-codirectional}:
        \begin{enumerate}
            \item Suppose $P \notin \mathcal{Q}_L \cup \mathcal{Q}_R \cup \mathcal{Q}_A \cup \mathcal{Q}_B$. We have $P' = P$. By Lemma~\ref{L:uncross-codirectional}, none of the edges of $P'$ are in $B_p$ or on $A$. By Lemma~\ref{L:opposite}, none of the edges of $P'$ are in $B$, so none of the edges of $P$ are on the boundary of $B_p$. On the other hand, $Q' \bigoplus Q$ consists only of edges in $B_p$ or on its boundary. It follows that if $P$ is below $Q$, then $P'$ is below $Q'$. Similarly, if $P$ is above $Q$, then $P'$ is above $Q'$. In both cases, $P'$ and $Q'$ do not cross.
            \item Suppose $P, Q \in \mathcal{Q}_L$. In both $P$ and $Q$ we replace $A$ with $rev(B)$ to get $P'$ and $Q'$. It follows that if $P$ is below $Q$, then $P'$ is below $Q'$. Similarly, if $P$ is above $Q$, then $P'$ is above $Q'$. In both cases $P'$ and $Q'$ do not cross
            \item Suppose $P \in \mathcal{Q}_L, Q \in \mathcal{Q}_B$. In $P$, we replace $A$ with $rev(B)$ to get $P'$. On the other hand, $Q$ must be below $P$. Since $B$ is below $P$ and $Q' \setminus Q$ consists of edges in $B$, we see that $Q'$ is below $P$ as well. By construction, $Q'$ is also on or below $B$, so $Q'$ is below $P'$ and does not cross it.
            \item Suppose $P, Q \in \mathcal{Q}_B$, and suppose without loss of generality that $P$ is below $Q$. Let $u$ be the first vertex at which $P$ crosses $B$ and let $v$ be the last. Let $s_P$ and $t_P$ be the endpoints of $P$. Then $P[s_P, u]$ and $P[v,t_P]$ are below $Q'$, so $P'$ is below $Q'$. 
        \end{enumerate}
        For the remaining four cases, suppose $P$ is left-to-right but $Q$ is right-to-left. We need to show that $P$ and $Q'$ do not cross each other more than $P$ and $Q$ cross each other:
        \begin{enumerate}
            \item[5.] Suppose $P \notin \mathcal{Q}_L \cup \mathcal{Q}_R \cup \mathcal{Q}_A \cup \mathcal{Q}_B$. 
            We have $P' = P$. As in case 1, none of the edges of $P$ are in $B_p$ or on the boundary of $B_p$. On the other hand, $Q' \bigoplus Q$ consists only of edges in $B_p$ or on its boundary. It follows that $P'$ and $Q'$ do not cross each other more than $P$ and $Q$ do.
            \item[6.] Suppose $P \in \mathcal{Q}_L$ and $Q \in \mathcal{Q}_R$. Path $P$ replaces $A$ with $rev(B)$ to get $P'$, and $Q$ replaces $B$ with $rev(A)$ to get $Q'$. Path $P$ contains $A$, so Lemma~\ref{L:opposite} implies that $P$ does not cross $B$ and so does not enter the interior of $B_p$. Similarly, $Q$ contains $B$ but does not enter the interior of $B_p$. This means that when we replace $P$ and $Q$ with $P'$ and $Q'$, the only vertices that could become crossing points or stop being crossing points are $x_p$ and $x_{p+1}$. But in fact $P$ contains $P_i[pred_i(x_p), succ_i(x_{p+1})] \supsetneq A$ and $Q$ contains $P_j[pred_j(x_{p+1}), succ_j(x_p)] \supsetneq B$, so both $x_p$ and $x_{p+1}$ are points at which $P$ and $Q$ cross and $P'$ and $Q'$ do not cross. Furthermore, no new crossings are added when we replace $P$ and $Q$ with $P'$ and $Q'$.
            \item[7.] Suppose $P \in \mathcal{Q}_L, Q \in \mathcal{Q}_A$. Note that $Q$ and $Q'$ are above $P_j$, while $P' \setminus P$ consists of edges in $B$, which is a subpath of $P_j$. 
            \item[8.] Suppose $P \in \mathcal{Q}_B, Q \in \mathcal{Q}_A$.
        \end{enumerate}
        All other cases are symmetric to these eight cases.
        Note that in case 6, $P'$ and $Q'$ cross each other fewer times than $P$ and $Q$, which is part of what we wanted to show.

\end{proof}

\end{document}